\documentclass[conference]{IEEEtran}

\usepackage{amsmath,amssymb}
\usepackage{subfigure}
\usepackage{graphicx,graphics,color,psfrag}
\usepackage{cite,balance}
\usepackage{caption}
\allowdisplaybreaks
\usepackage{algorithm}
\usepackage{accents}
\usepackage{amsthm}
\usepackage{bm}
\usepackage{algorithmic}
\usepackage[english]{babel}
\usepackage{multirow}
\usepackage{enumerate}
\usepackage{cases}
\usepackage{stfloats}
\usepackage{dsfont}
\usepackage{color,soul}
\usepackage{amsfonts}
\usepackage{cite,graphicx,amsmath,amssymb}
\usepackage{subfigure}
\usepackage{fancyhdr}
\usepackage{hhline}
\usepackage{graphicx,graphics}
\usepackage{array,color}
\usepackage{booktabs}
\usepackage{diagbox}

\usepackage{amsmath}

\usepackage{lipsum}
\usepackage{mwe}
\newtheorem{prop}{Proposition}


\newcounter{parentnumber}

\theoremstyle{definition}
\newtheorem{definition}{Definition}
\newtheorem{remark}{Remark}
\newtheorem{lemma}{Lemma}
\include{header}

\begin{document}
	
	\title{\Huge Realizing Multi-Point Vehicular Positioning via Millimeter-wave Transmission}
	
\author{\IEEEauthorblockN {Zezhong Zhang\IEEEauthorrefmark{1},
		Seung-Woo Ko\IEEEauthorrefmark{2}, Rui Wang\IEEEauthorrefmark{3}, and 
		Kaibin Huang\IEEEauthorrefmark{1} }
	
	\IEEEauthorblockA{
		\IEEEauthorrefmark{1}Department of Electrical and Electronic Engineering, The University of Hong Kong\\
		\IEEEauthorrefmark{2}Division of Electronics and Electrical Information Engineering, Korea Maritime and Ocean University\\
		\IEEEauthorrefmark{3}Department of Electrical and Electronic Engineering, Southern University of Science and Technology\\
		Email:  zzzhang@eee.hku.hk, swko@kmou.ac.kr, wang.r@sustech.edu.cn, huangkb@eee.hku.hk}}

	% make the title area
	\maketitle
	
	\begin{abstract}	
		Multi-point detection of the full-scale environment is an important issue in autonomous driving. The state-of-the-art positioning technologies (such as RADAR and LIDAR) are incapable of real-time detection without line-of-sight. To address this issue, this paper presents a novel multi-point vehicular positioning technology via \emph{millimeter-wave} (mmWave) transmission that exploits multi-path reflection from a \emph{target vehicle} (TV) to a \emph{sensing vehicle} (SV), which enables the SV to fast capture both the shape and location information of the TV in \emph{non-line-of-sight} (NLoS) under the assistance of multi-path reflections. A \emph{phase-difference-of-arrival} (PDoA) based hyperbolic positioning algorithm is designed to achieve the synchronization between the TV and SV. The \emph{stepped-frequency-continuous-wave} (SFCW) is utilized as signals for multi-point detection of the TVs. 
		Transceiver separation enables our approach to work in NLoS conditions and achieve much lower latency compared with conventional positioning techniques. 
		%In NLoS conditions, leveraging the geometry relations due to reflection,  we establish a system of equations and estimate the TV's location with common points in each retrieved virtual position. In the end, the realistic simulation of vehicular positioning with medium range in rural scenarios demonstrates the effectiveness of the proposed mmWave-based multi-point positioning techniques. 
	\end{abstract}
	
	\section{Introduction}
	
	Autonomous driving has grown into a reality with rapid progress of diverse technologies~\cite{HeathAutoDrivingIntro}. Much effort and investment have been devoted by top automobile companies (Tesla, BMW) and Internet companies (Google, Baidu), leading to significant achievements toward  commercialization. 
	One remaining challenge to attain full autonomous driving is to accurately recognize vehicles nearby, termed \emph{vehicular positioning}~\cite{VehicularPositioning}.
	Compared to conventional single-point positioning approaches~\cite{PPP,MobileCentric}, vehicular positioning is required to accurately estimate high-resolution positioning information including location, size, and shape of \emph{target vehicle} (TV). We call it \emph{multi-point vehicular positioning} in this paper.

	\subsection{Positioning}

	Single-point positioning has been widely used in the area of mobile positioning, which determines the object's position as a single point since a mobile device is typically small enough to mark a representative point.  
	The most common way is to use a built-in \emph{Global Positioning System} (GPS) receiver but its usage is limited since GPS signals are  frequently blocked in urban environments. On the other hand, single-point positioning is unsuitable for autonomous driving due to the fact that a vehicle is too big to represent a single point. The lack of shape and size information may cause fatal accidents~\cite{Accident}. One can claim to be able to deliver shape and size information separately through a reliable communication link but it needs additional efforts to align them with the estimated point perfectly.

	To overcome the limitation of single-point positioning, there have been efforts in the area of multi-point positioning. {Passive multi-point positioning} detects the natural radiation from the target objects without emitting discernible radiations. Infrared sensors and cameras are representative ones, which easily retrieve the positioning information in a fully covert manner \cite{PassiveMMW}.
	On the other hand, \emph{active multi-point positioning} techniques illuminate the target and detect its position with data extracted from reflected signals or lights~\cite{PassiveMMW,BackPropagation}. A \emph{RAdio-Detection-And-Ranging} (RADAR) system is the most popular one in surface and subsurface detections, which is implemented by impulse waveforms and \emph{continuous-waves} (CWs) with different frequency modulations~\cite{SFCW2016}. 
	\emph{LIght-Detection-And-Ranging} (LIDAR), another new emerging active positioning technique, utilizes narrow laser beams for positioning and uses the scanning mirror for fast scanning~\cite{LIDARscanning}.
	
	More importantly, a fatal drawback of the above techniques is that they are only capable of detecting the vehicles in LoS since the corresponding mediums cannot penetrate a large solid blockage in the road such as a truck or a bus. However, the disability of detection in NLoS  results in severe safety issues because many car accidents happen when the drivers are unaware of the environment. 
	Therefore, designing techniques for positioning in NLoS is an urgent task for applications of autonomous driving in real-life. In~\cite{Kaifeng}, a NLoS positioning technique is developed to estimate a vehicle's position by exploiting the geometry information of multi-path signal transmissions. It is also possible to infer the vehicle's size and shape if the vehicle equips multiple antenna clusters, but the resultant resolution is low.

	\subsection{Main Contributions}
	The contributions of this work are summarized as follows.

	\noindent 1) {\bf Real-time Positioning:}  The proposed multi-point vehicular positioning is a one-way simultaneous multi-antenna transmissions from the TV to the SV. Compared to conventional multi-point positioning techniques requiring time-consuming scanning process, our technique is able to achieve ultra-low latency and realize real-time multi-point positioning. 
	
	{\noindent \noindent 2) {\bf Synchronization:} In the above one-way simultaneous transmission system, 
		a prerequisite for the multi-point positioning is to know a clock synchronization gap between the SV and the TV. 
		To this end, we design a novel synchronization algorithm based on multiple \emph{Phase-Difference-of-Arrival} (PDoA) information obtained by transmitting \emph{signature waveforms} (SWs) from representative antennas. The PDoA information leads to constructing a system of equations following a hyperbolic geometry.

		{\noindent \noindent 3) {\bf Positioning in LoS and NLoS:} In LoS case, A FFT-based signal processing technique is used to retrieve the location information of all transmit antennas, i.e., the TV's position, based on the received signals. In NLoS, the SV uses multiple \emph{mirror vehicles} (MVs) as reflectors of the TV's transmitted mmWave signal to estimate the TV's position. 	By the aid of specular reflection, the SV is able to achieve the NLoS TV's multi-point positioning if the MV's locations are given but unknown in practice. This difficulty is overcome by exploiting the geometry relation between the vehicle's reflected and \emph{real positions} (RPs).  
			
			%the location information of representative points on each VP obtained during the synchronization. These representative points form another geometry relation due to the fact that they are originated from common points on RP. The MVs' locations can be inferred by exploiting such geometry property, and the information of MVs leads to transforming VPs into RP.
			
			%\textcolor{red}{The remainder of this paper is organized as follows. In Section II, the system model and transmission model are introduced. In Section III and IV, the proposed imaging technique is elaborated and analyzed in the NLoS condition. The realistic simulation results are presented in Section V, and the conclusion follows in Section VP.}
			\section{System Model}
			
			\begin{figure}[t]
				\centering
				\includegraphics[%height=230pt, 
				width=150pt]{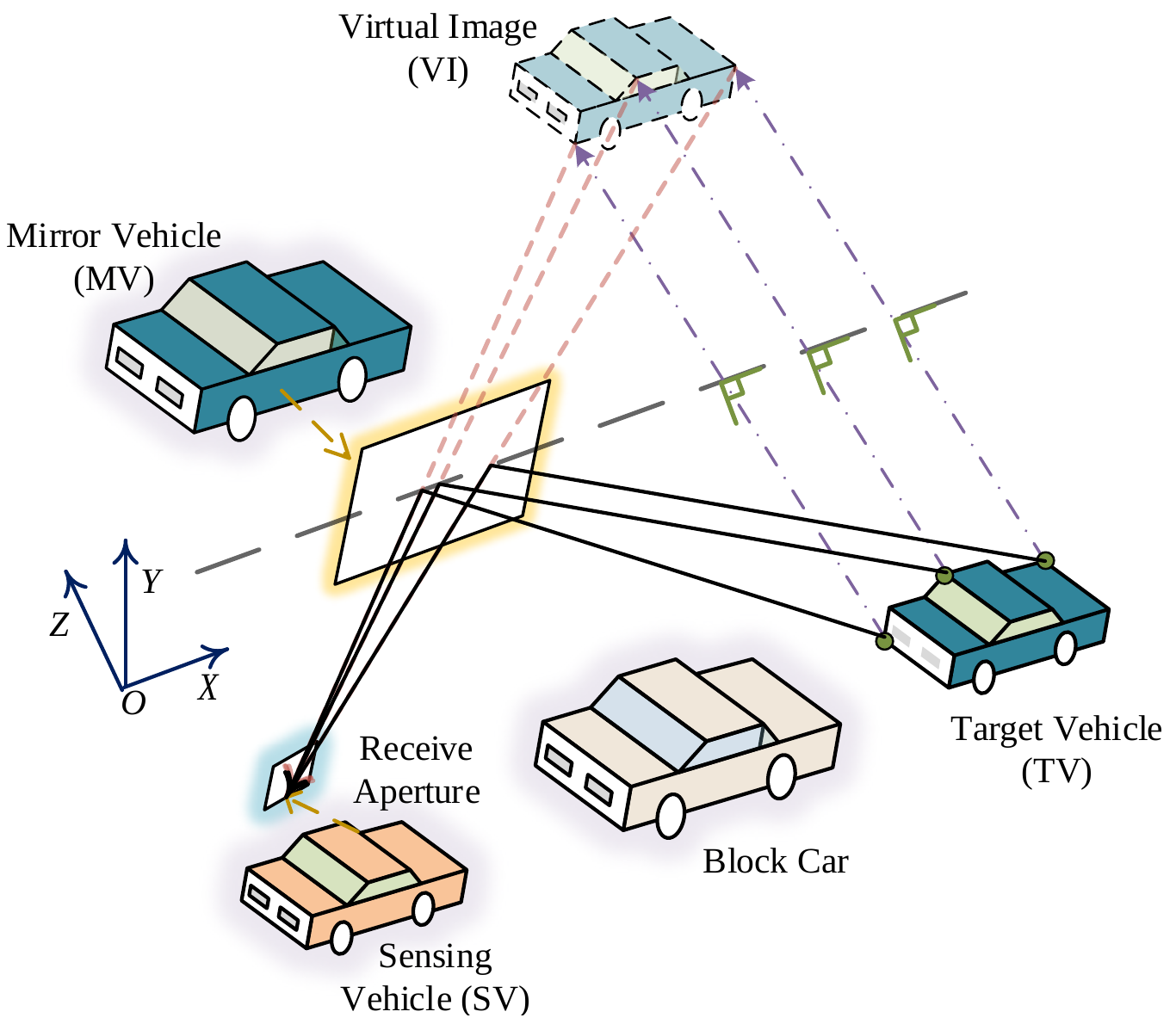}
				\caption{Illustration of signals reflection by mirror vehicles.}\label{mirror}
			\end{figure}
			
			Consider the scenario with multiple vehicles located on the road. Each vehicle is equipped with an antenna array around the vehicle body, which can generally represent its shape. We adopt the wide-band positioning system using \emph{stepped-frequency-continuous-wave} (SFCW)~\cite{SFCW2016} at mmWave spectrums as a waveform for each {transmit}  antenna. It comprises multiple CW signals with different frequencies, each of which is separated by a certain amount. 
			%The key principle of SFCW-based positioning is to embed position information into these spatial-and-frequency domain signals' propagations, which can be converted by Fourier transform elaborated in the sequel. 
			%{Considering radio-wave?s propagation with constant light speed, 
			%the image information is completely embedded into these spatial-and-frequency domain signals. Thus the 3-D image can be recovered by typical spatial domain-frequency domain conversion, Fourier transform in this work, with the collected data over different azimuth and range.}
			% Once the phase information of the received signals is correctly detected, which is directly related to the travel time, it is possible  to recover the image of this vehicle elaborated in the sequel.  
			
			In mmWave bands, due to high attenuation loss and sparse distribution of scatters, most signal propagations follow LoS especially when vehicles are dispersive. However, a vehicle's metal body with a favorable reflection property makes it possible to propagate signals even in NLoS. Fig.~\ref{mirror} graphically illustrates the above property such that the SV tries to detect a TV by receiving the TVs' signals, but the SV is blocked by other vehicles and cannot see the TV directly. One alternative is to exploit nearby MVs as reflectors of the signals. With multiple MVs, the SV is able to obtain the position of the TV. Without loss of generality, we assume that the {receive}  aperture is located at the SV's left side, and $X$, $Y$, and $Z$-axes represent its moving direction, height, and width, respectively.

			\subsection{Signal Model}\label{SignalModel}
			
			Two types of TV transmissions are considered depending on different purposes. The first is the SFCW transmission enabling to obtain TV's position at the SV, and the second is a  \emph{signature waveform} (SW)  transmission~\cite{SignatureWaveform} to compensate the SV-TV synchronization gap.

			\subsubsection{SFCW Transmission}
			All TV's antennas  simultaneously broadcast the same SFCW waveform denoted by ${\boldsymbol{s}}(t)$  as
			\begin{align}\label{transmitSig}
			{\boldsymbol{s}}(t)= \left[\exp(j{2\pi {f_1}{t}}), ..., \exp(j{2\pi {f_K}{t}}) \right]^T,
			\end{align}
			where $\{f_k\}_{k=1}^K$ represents the set of frequencies with constant gap $\Delta$ such that $f_k=f_1+(k-1)\Delta$ for $k=1, \cdots, K$. The received signal at the SV's antenna $m$  is given as 
			\begin{align}\label{Eq:TotReceivedSig}
			{\boldsymbol r}_{m}(t)=\sum_{\ell=0}^L  {\boldsymbol r}_{m}^{(\ell)}(t),
			\end{align}
			where ${\boldsymbol r}_{m}^{(\ell)}(t)$ denotes the signal reflected by the  $\ell$-th MV as 
			\begin{align}\label{receiveSig}
			{\boldsymbol r}_{m}^{(\ell)}(t)&=\Gamma^{(\ell)}\sum_{n=1}^N {\boldsymbol r}_{n,m}^{(\ell)}(t)=\Gamma^{(\ell)}\sum_{n=1}^N {\boldsymbol s}(t+\sigma-\tau_{n,m}^{(\ell)}). 
			\end{align}
			Here, $\Gamma^{(\ell)}$ is the complex reflection coefficient given as $\Gamma^{(\ell)}=|\Gamma^{(\ell)}|\exp(j \angle \Gamma^{(\ell)})$\footnote{The reflection planes are sides of vehicles, which are usually plane and smooth. In addition, the limited size of the receive aperture makes the incident angles of the signals almost the same. As a result, the $\Gamma^{(\ell)}$ can be well approximated as a  constant regardless of antennas (see e.g., \cite{SFCW2016,Antenna}).},  $N$ is the number of TV's antennas, $\sigma$ is the TV-SV synchronization gap (in sec), and $\tau_{n,m}^{(\ell)}$ is the signal travel time from TV's antenna $n$ to SV's antenna $m$ proportional to the propagation distance $d_{n,m}$, i.e., $d_{n,m}=c\cdot \tau_{n,m}^{(\ell)}$ where $c=3\cdot 10^8$ (m/sec) is the speed of light. Note that signal path $\ell=0$ represents the LoS path of which the reflection coefficient $\Gamma^{(0)}$ is one. Last, we assume that the signals reflected by different MVs come from different directions, facilitating to differentiate signals from different MVs according to the \emph{angle-of-arrival} (AoA). In other words, we can decompose \eqref{Eq:TotReceivedSig} into individual ${\boldsymbol r}_{m}^{(\ell)}(t)$ as
			\begin{align}\label{Eq:MatrixForm}
			{\mathbf{R}}_{m}(t)=\left[ {\boldsymbol{r}}_{m}^{(0)}(t), {\boldsymbol{r}}_{m}^{(1)}(t), \cdots,  {\boldsymbol{r}}_{m}^{(L)}(t)\right].
			\end{align}

			The synchronization gap $\sigma$ is an unknown parameter the SV attempts to estimate. Assuming the estimiated gap is $\tilde{\sigma}$, the received signal \eqref{Eq:MatrixForm} is demodulated by multiplying ${\bf D}= \mathsf{diag}\{{\boldsymbol{s}}(t+\tilde{\sigma})^{H}\}$
			\begin{align}
			{\mathbf{Y}}_m=\left[ {\mathbf{y}}_{m}^{(0)}, {\mathbf{y}}_{m}^{(1)}, \cdots,  {\mathbf{y}}_{m}^{(L)}\right]
			={\bf D}{\mathbf{R}}_{m}(t),
			\end{align}
			where
			${\mathbf{y}}_{m}^{(\ell)} = {\bf D} {\mathbf{r}}_{m}^{(\ell)}(t)= \left[y_m^{\ell,1}, y_m^{\ell,2},..., y_m^{\ell,K}\right]^T$ with
			\begin{align}\label{demodulated}
			y_m^{\ell,k} &= \Gamma^{(\ell)}\sum_{n=1}^N{{\exp\left[ { j2\pi {f_k}(\sigma-\tilde{\sigma}-{\tau_{n,m}^{(\ell)}} )  } \right]}}. 
			\end{align}
			%It is shown that the demodulation signal $y_m^{\ell,k}$ \eqref{demodulated} totally depends on $\{\tau^{(\ell)}_{n,m}\}$ 
			%if the synchronization gap is perfectly estimated ($\sigma=\tilde{\sigma}$). Otherwise, it is distorted. 

			%\begin{remark}[Feasible Distance of SFCW]\label{feasibleDis}{Due to the periodicity of phases, the usage of the SFCW should be limited by $R_{\max}=\frac{c}{\Delta}$ (in meters) to avoid ambiguity \cite{SFCW2016} where 
			%the frequency gap $\Delta$ should be larger than $10$MHz due to the state-of-the-art hardware limitation, the resultant feasible distance is up to $30$m approximately. }
		%	\end{remark}

			%\begin{remark}[Relative Velocity Compensation]
			%	Due to the relative velocity between the TV and SV, Doppler shift may cause severe main lobe spreading and the loss of 
			%	\emph{Signal-to-noise-ratio} (SNR). Some existing technologies have been used for velocity compensation and Doppler shift mitigation. For example, by detecting the main lobe width of the SFCW's range profile, a velocity estimation and compensation approach with high accuracy is proposed in~\cite{Doppler}. Therefore, we assume perfect velocity compensation in this paper.
			%\end{remark}
			
			\subsubsection{SW Transmission}	
			Two representative antennas, $\mathsf{a}$ and $\mathsf{b}$, are selected in the array of the TV with coordinates  $\mathbf{x}_{\mathsf{a}} = (x_{\mathsf{a}}, y_{\mathsf{a}}, z_{\mathsf{a}})$ and $\mathbf{x}_{\mathsf{b}} = (x_{\mathsf{b}}, y_{\mathsf{b}}, z_{\mathsf{b}})$, respectively.  To estimate $\mathbf{x}_{\mathsf{a}}$ and $\mathbf{x}_{\mathsf{b}}$ as  an intermediate step for the  synchronization, each of them simultaneously transmits SWs comprising two CWs with different frequencies as
			\begin{align}\label{Signature_Waveform}
			\boldsymbol{s}_{\mathsf{a}}(t)=&[\exp(j2\pi f_{\mathsf{a}} t), \exp(j2\pi (f_{\mathsf{a}}+\Delta) t)]^T,\nonumber\\
			\boldsymbol{s}_{\mathsf{b}}(t)=&[\exp(j2\pi f_{\mathsf{b}} t), \exp(j2\pi (f_{\mathsf{b}}+\Delta) t)]^T,
			\end{align}
			where $f_{\mathsf{a}}$ and $f_{\mathsf{b}}$ represent the SWs' frequencies originated from the antennas $\mathsf{a}$ and $\mathsf{b}$ respectively, and $\Delta$ is the frequency separation for another CW in each SW. 
			The frequency bands for the SWs do not  overlap to each other and are different from that of the SFCW,
			namely, $f_{\mathsf{a}}<f_{\mathsf{a}}+\Delta<f_{\mathsf{b}}<f_{\mathsf{b}}+\Delta<f_1$, 
			enabling to receive and demodulate the SWs independently from the SFCW without interference. 
			Similar to \eqref{Eq:MatrixForm}, the received signals at the SV's antenna $m$ are 
			expressed as
			\begin{align}
			{\boldsymbol A}_{m}\!(t)\!=\![{\boldsymbol a}_{m}^{(1)}\!(t), \cdots,{\boldsymbol a}_{m}^{(L)}\!(t)], 
			{\boldsymbol B}_{m}\!(t)\!=\![{\boldsymbol b}_{m}^{(1)}\!(t), \cdots, {\boldsymbol b}_{m}^{(L)}\!(t)], \nonumber
			\end{align}
			where 
			${\boldsymbol a}_{m}^{(\ell)}(t)\!\!=\!\!\Gamma^{(\ell)} {\boldsymbol s}_{\mathsf{a}}(t\!+\!\sigma\!-\!\tau_{\mathsf{a},m}^{(\ell)})$ and 
			${\boldsymbol b}_{m}^{(\ell)}(t)\!=\!\Gamma^{(\ell)} {\boldsymbol s}_{\mathsf{b}}(t\!+\!\sigma\!-\!\tau_{\mathsf{b},m}^{(\ell)})$. 
			By multiplying ${\bf D}_{\mathsf{a}}\!=\! \mathsf{diag}\!\left\{\!\boldsymbol{s}_{\mathsf{a}}(t)^H\!\right\}$ and 
			${\bf D}_{\mathsf{b}}\!=\! \mathsf{diag}\!\left\{\!\boldsymbol{s}_{\mathsf{b}}(t)^H
			\!\right\}$ respectively, ${\boldsymbol A}_{m}(t)$ and ${\boldsymbol B}_{m}(t)$ can be demodulated as
			${\bf D}_{\mathsf{a}}{\mathbf{A}}_{m}(t)\!\!=\!\!\left[\! {\boldsymbol{\alpha}}_{m}^{(0)}, {\boldsymbol{\alpha}}_{m}^{(1)}, \cdots,  {\boldsymbol{\alpha}}_{m}^{(L)}\!\right]$ and $
			{\bf D}_{\mathsf{b}}{\mathbf{B}}_{m}(t)\!\!=\!\!\left[\! {\boldsymbol{\beta}}_{m}^{(0)}, {\boldsymbol{\beta}}_{m}^{(1)}, \cdots,  {\boldsymbol{\beta}}_{m}^{(L)}\!\right]$, where
			\begin{align} 
			{\boldsymbol{\alpha}}_{m}^{(\ell)} \!&
			=\!\Gamma^{(\ell)}\!\!\left[\!\exp\! \left(\! j{2\pi {f_{\mathsf{a}}}\!(\sigma\!\!-\!\!{\tau _{\mathsf{a},m}^{(\ell)}})} \!\right), 
			\exp\! \left(\! j{2\pi (f_{\mathsf{a}}\!\!+\!\!\Delta)({\sigma\!\!-\!\!{\tau _{\mathsf{a},m}^{(\ell)}}})} \!\right)\!\right],\nonumber\\
			{\boldsymbol{\beta}}_{m}^{(\ell)}  \!&=\!\Gamma^{(\ell)}\!\!\left[\!\exp\! \left(\! j{2\pi {f_{\mathsf{b}}}\!( {\sigma\!\!-\!\!{\tau _{\mathsf{b},m}^{(\ell)}}})} \!\right), 
			\exp\! \left(\! j{2\pi (f_{\mathsf{b}}\!\!+\!\!\Delta)( {\sigma\!\!-\!\!{\tau _{\mathsf{b},m}^{(\ell)}}})} \!\right)\!\right].\nonumber
			\end{align}

			%with the component $\alpha_{n_1, m}^{\ell,k}$ being

			%\begin{remark}[Synchronization Gap Range]{Due to the constant frequency gap $\Delta$, SFCW and SWs have the same period as $\frac{1}{\Delta}$. In other words, the range of the synchronization gap $\sigma$ to be considered is limited from zero  to $\frac{1}{\Delta}$, i.e., $0\leq \sigma< \frac{1}{\Delta}$, translated into the feasible distance $\frac{c}{\Delta}$ (m) considering the speed of light $c$.  }
			%\end{remark}
			
			%	{\zzz
			%    \begin{remark}
			%    	Synchronization can be achieved only when the SFCW frequency gap $\Delta$ is divisable by the frequency separation $\Delta_f$ of a SW pair.
			%    \end{remark}
			%\begin{proof}
			%	Please refer to Appendix A.
			%\end{proof}
			
		}
		\subsection{Problem Formulations}
		
		To establish the direct relation between the propagation distance and the phase extracted from $y_m^{\ell,k}$ \eqref{demodulated}, the SV aims at compensating the clock synchronization gap $\sigma$, namely,
		\begin{align}\label{SyncProblemFormulation}\tag{E1}
		\bigcup_{\ell=1}^L\{{\boldsymbol{\alpha}}_{m}^{(\ell)}, {\boldsymbol{\beta}}_{m}^{(\ell)}\}_{m=1}^M \Longrightarrow \tilde{\sigma}=\sigma, %\quad \forall \ell,
		\end{align}
		where $M$ is the number of receive antennas deployed in the SV. Assume that the synchronization is made. The  $y_m^{\ell,k}$ \eqref{demodulated} is then rewritten by the following surface integral form:
		\begin{align}\label{demodulated_sync}
		y_m^{\ell,k} \underset{(\sigma=\tilde{\sigma})}{\Longrightarrow} \Gamma^{(\ell)}\int {\mathsf{I}_\mathbf{x^{(\ell)}}{\exp\left( - j2\pi\frac{f_k}{c} D(\mathbf{x}, \mathbf{p}_{m}) \right)}d{\mathbf{x}}},
		\end{align}
		where $\mathsf{I}_\mathbf{x^{(\ell)}}$ is an indicator to become one if a TX antenna exists on point $\mathbf{x}$, which is symmetric to the point ${\mathbf{x^{(\ell)}}}$ w.r.t. the surface of MV $\ell$, and zero
		otherwise, and $D(\mathbf{x}, \mathbf{p}_{m})$ represent the total propagation distance between $\mathbf{x}$ and the location of 
		RX antenna $m$ denoted by $\mathbf{p}_m$. Estimating $\{\mathsf{I}_{\mathbf{x^{(\ell)}}}\}$ is equivalent to detecting TV's position $\ell$, namely, 
		\begin{align}\label{ProblemDefinition1}\tag{E2}
		\bigcup_{k=1}^K\bigcup_{m=1}^M y_m^{\ell,k} \Longrightarrow\{\mathsf{I}_\mathbf{x^{(\ell)}}\}. 
		\end{align}
		It is worth noting that in case of LoS path ($\ell=0$), the distance $D(\mathbf{x}, \mathbf{p}_m)$ is the direct distance between $\mathbf{x}$ and $\mathbf{p}_m$. Thus, the position of real TV is directly detected.   
		In NLoS case, $D(\mathbf{x}, \mathbf{p}_m)$ corresponds to the total distance from $\mathbf{x}$ via MV $\ell$ to $\mathbf{p}_m$. Since SV has no priori information of MV's location, the detected position could be different from the RP due to reflection, which is called a \emph{virtual position} (VP) (see Fig. \ref{mirror}). It is necessary to map multiple VPs into the RP, namely,
		\begin{align}\label{ProblemDefinition2}\tag{E3}
		\bigcup_{\ell=0}^L \{\mathsf{I}_\mathbf{x^{(\ell)}}\} \Longrightarrow  \{\mathsf{I}_\mathbf{x}\}.
		\end{align}
		
		%In summary, the proposed multi-point vehicular positioning technique follows three steps: 1) synchronization between the TV and the SV by solving \ref{SyncProblemFormulation}, 2) the multi-point detection of position $\ell$  by solving \ref{ProblemDefinition1}, and 3) mapping multiple VPs by solving \ref{ProblemDefinition2}. 

		\section{Multi-Point Vehicular Positioning}\label{sec:LoS}
		In this section, we consider a case where a LoS path between the TV and the SV exists, making it reasonable to ignore other NLoS paths due to the significant power difference between LoS and NLoS paths.
		Thus only synchronization and multi-point positioning steps are needed, which are explained in detail, following overview, algorithm description and performance analysis.
		
	\subsection{Synchronization}\label{sec:Synchronization}
%\begin{figure}
%	\centering
%	\includegraphics[%height=220pt, 
%	width=250pt]{Figures/syn.pdf}
%	\caption{Illustration of the synchronization approach.}\label{Syn}
%\end{figure}
\subsubsection{Overview} We apply the technique of \emph{phase-difference-of-arrival} (PDoA) based localization~\cite{PDoA}  to compensate the synchronization gap, which is illustrated in the following. 
Consider the SWs from the representative antenna ${\mathsf{a}}$ first. The received SWs at the receive antenna $m$ through the signal path reflected by the $\ell$-th MV are given as
\begin{align}
{\boldsymbol{\alpha}}_{m}^{(\ell)} \!&
=\!\Gamma^{(\ell)}\!\!\left[\!\exp\! \left(\! j{2\pi {f_{\mathsf{a}}}\!(\sigma\!\!-\!\!{\tau _{m}^{(\ell)}})} \!\right), 
\exp\! \left(\! j{2\pi (f_{\mathsf{a}}\!\!+\!\!\Delta)({\sigma\!\!-\!\!{\tau _{m}^{(\ell)}}})} \!\right)\!\right],\nonumber
\end{align}
where $\tau_m^{(\ell)}$ represents the flight time, and the indices of the transmit antennas are omitted for brevity. At the SV's antenna $m$, the phase difference between the two components is calculated as $\eta_{ m}^{(\ell)}=2\pi {\Delta}({\tau _{m}^{(\ell)}} -\sigma)$. Note that $\sigma$ is the same for signals from different paths. Recalling the relation between the propagation distance $d_m^{(\ell)}$ and $\tau_m^{(\ell)}$, i.e., $d_m^{(\ell)}=\tau_m^{(\ell)}\cdot c$ with light speed $c$, the synchronization gap $\sigma$ is given as
\begin{align}\label{Sigma}
\sigma=\tau_m^{(\ell)}-\frac{\eta_m^{(\ell)}}{2\pi\Delta}=\frac{d_m^{(\ell)}}{c}-\frac{\eta_m^{(\ell)}}{2\pi\Delta}.
\end{align}
Moreover, according to the geometric relation between the $\ell$-th VP and the TV shown in Fig.~\ref{mirror}, the propagation distance can be presented as $d_m^{(\ell)} = \left\| {{{\mathbf{p}}_{{m}}} - {{\mathbf{x}}_{\mathsf{a}}^{(\ell)}}} \right\|$, where $\mathbf{x}_{\mathsf{a}}^{(\ell)}$ is the location of the antenna $\mathsf{a}$ on the VP $\ell$, which are different from the real location denoted by  $\mathbf{x}_{\mathsf{a}}$ in case of NLoS.
Because the locations of all SV's antennas $\{\mathbf{p}_m\}$ are given, the synchronization problem is translated to find the location  $\mathbf{x}_{\mathsf{a}}^{(\ell)} = (x_{\mathsf{a}}^{(\ell)}, y_{\mathsf{a}}^{(\ell)}, z_{\mathsf{a}}^{(\ell)})$.
Specifically, let $\mathsf{F}_{i}(\mathbf{x}_{n}^{(\ell)})$ denote the propagation distance difference from the antenna $n$ on the VP $\ell$ to the SV's antennas $m$ and $1$ as 
\begin{align}\label{FuncDefine}
\mathsf{F}_{m}(\mathbf{x}_{\mathsf{a}}^{(\ell)}) =& \left\| {{{\mathbf{p}}_{{m}}} \!-\! {{\mathbf{x}}_{\mathsf{a}}^{(\ell)}}} \right\| \!-\! \left\| {{{\mathbf{p}}_{{1}}} \!-\! {{\mathbf{x}}_{\mathsf{a}}^{(\ell)}}} \right\|\nonumber\\
=&d_m^{(\ell)}\!-\!d_1^{(\ell)}
\!\overset{(a)}{=}\! c\frac{\left(\! \eta _{m}^{(\ell)} \!-\! \eta _{1}^{(\ell)} \!\right)}{2\pi \Delta }, \ m\!=\!2,\cdots, M,
\end{align}
where (a) follows from  \eqref{Sigma}. In \eqref{FuncDefine}, there are $M-1$ equations with three unknowns $x_{\mathsf{a}}^{(\ell)}$, $y_{\mathsf{a}}^{(\ell)}$, and $z_{\mathsf{a}}^{(\ell)}$.
%\begin{prop} [Synchronization Feasibility Condition]\label{Prop1}
%	\emph{To compensate the synchronization gap, at least four SV's antennas are required: $M\geq 4$.}
%\end{prop}

%\begin{remark}[Sampling Requirement]\label{remark1}{
%The synchronization procedures above are based on the assumption that the phase gap estimated at each two adjacent receive antennas is no larger than $\frac{\pi}{2}$. Thus the distance between each two adjacent receive antenna at the SV needs to satisfy $\Delta d < \frac{c}{{2\Delta }}$. This requirement can be easily satisfied in practice.}
%\end{remark}

\subsubsection{Algorithm Description}	
One challenge to solve \eqref{FuncDefine} is a nonlinearity of $\mathsf{F}_{m}(\mathbf{x}_{\mathsf{a}}^{(\ell)})$ following a hyperbolic geometry. To overcome the difficulty, we use the the Gauss-Newton method~\cite{Newton}. 
Assume that the initial value $\tilde{\mathbf{x}}_{\mathsf{a}}^{(\ell)}$ is a good estimator of $\mathbf{x}_{\mathsf{a}}^{(\ell)}$. Then, $\mathsf{F}_{m}(\mathbf{x}_{\mathsf{a}}^{(\ell)})$ can be approximated as 
\begin{align}\label{Approximation}
\mathsf{F}_{m}({{\mathbf{x}}_{\mathsf{a}}^{(\ell)}}) \approx {\mathsf{F}_{m}}({\tilde{\mathbf{x}}_{\mathsf{a}}^{(\ell)}})+\frac{\partial {\mathsf{F}_{m}}({\tilde{\mathbf{x}}_{\mathsf{a}}^{(\ell)}})}{{\partial {{\mathbf{x}}_{\mathsf{a}}^{(\ell)}}}}\boldsymbol{h},
\end{align}
where $\boldsymbol{h}=[h_x,h_y,h_z]^T$. Plugging \eqref{Approximation} into \eqref{FuncDefine} gives the following linear of equations: 
\begin{align}\label{LinearSystemofEquation}\tag{E4}
{\bf{G}}(\tilde{\mathbf{x}}_{\mathsf{a}}^{(\ell)})\boldsymbol{h} = {\bf{b}}(\tilde{\mathbf{x}}_{\mathsf{a}}^{(\ell)}),
\end{align}
where 
\begin{align}
{\bf{G}}(\tilde{\mathbf{x}}_{\mathsf{a}}^{(\ell)})=\left[ {\begin{array}{*{20}{c}}
	{\frac{{\partial {\mathsf{F}_{2}}({\tilde{\mathbf{x}}_{\mathsf{a}}^{(\ell)}})}}{{\partial x_{\mathsf{a}}^{(\ell)}}}}&{\frac{{\partial {\mathsf{F}_{2}}({\tilde{\mathbf{x}}_{\mathsf{a}}^{(\ell)}})}}{{\partial y_{\mathsf{a}}^{(\ell)}}}}&{\frac{{\partial {\mathsf{F}_{2}}({\tilde{\mathbf{x}}_{\mathsf{a}}^{(\ell)}})}}{{\partial z_{\mathsf{a}}^{(\ell)}}}}\\
	{...}&{...}&{...}\\
	{\frac{{\partial {\mathsf{F}_{M}}({\tilde{\mathbf{x}}_{\mathsf{a}}^{(\ell)}})}}{{\partial x_{\mathsf{a}}^{(\ell)}}}}&{\frac{{\partial {\mathsf{F}_{M}}({\tilde{\mathbf{x}}_{\mathsf{a}}^{(\ell)}})}}{{\partial y_{\mathsf{a}}^{(\ell)}}}}&{\frac{{\partial {\mathsf{F}_{M}}({\tilde{\mathbf{x}}_{\mathsf{a}}^{(\ell)}})}}{{\partial z_{\mathsf{a}}}}}
	\end{array}} \right], \nonumber
\end{align}
\begin{align}
{\bf{b}}(\tilde{\mathbf{x}}_{\mathsf{a}}^{(\ell)})=
\left[ {\begin{array}{*{20}{c}}
	{c\frac{{\left( {{\eta _{2}-\eta _{1}} } \right)}}{{2\pi \Delta }} - {\mathsf{F}_{2}}({\tilde{\mathbf{x}}_{\mathsf{a}}^{(\ell)}})}\\
	{...}\\
	{c\frac{{\left( {{\eta _{M}-\eta _{1}} } \right)}}{{2\pi \Delta }} - {\mathsf{F}_{M}}({\tilde{\mathbf{x}}_{\mathsf{a}}^{(\ell)}})}
	\end{array}} \right].
\end{align}
The estimated value of $\boldsymbol{h}$ can thus be given as
\begin{align}\label{Solution}
\boldsymbol{h} = {\left( {{{\bf{G}}(\tilde{\mathbf{x}}_{\mathsf a}^{(\ell)})^T}{\bf{G}}(\tilde{\mathbf{x}}_{\mathsf{a}}^{(\ell)})} \right)^{ - 1}}{{\bf{G}}(\tilde{\mathbf{x}}_{\mathsf{a}}^{(\ell)})^T}{\bf{b}}(\tilde{\mathbf{x}}_{\mathsf{a}}^{(\ell)}).
\end{align}
The estimated location of ${{\mathbf{x}}_{\mathsf{a}}^{(\ell)}}$ is then updated as
\begin{align}\label{Result}
{\tilde{\mathbf{x}}_{\mathsf{a}}^{(\ell)}} \longleftarrow {\tilde{\mathbf{x}}_{\mathsf{a}}^{(\ell)}} + \boldsymbol{h}. 
\end{align}
By repeating the update procedure several times, $\tilde{\mathbf{x}}_{\mathsf{a}}^{(\ell)}$ converges to ${\mathbf{x}}_{\mathsf{a}}^{(\ell)}$. Besides, ${\mathbf{x}}_{\mathsf{b}}^{(\ell)}$ can be derived in the same way, which helps locate the TV in Sec. \ref{sec2:Common-point}.
Substituting ${\mathbf{x}}_{\mathsf{a}}^{(\ell)}$ into \eqref{Sigma} leads to the accurate synchronization gap $\sigma$. Moreover, since the synchronization gap is not affected by the MVs, it helps the SV to judge whether the signals of different AOAs, are from the same TV.

\begin{remark}[Initial Value Selection]{The Gauss-Newton method sometimes converges to a local optimal point due to the wrong selection of the initial  ${\tilde{\mathbf{x}}_{\mathsf{a}}^{(\ell)}}$~\cite{Gauss_Newton}. It is recommended to use the solution satisfying any three equations among \eqref{FuncDefine} as its initial selection, of which the convergence of the global optimal is verified by simulation. }
\end{remark}

\begin{remark}[Synchronization with Phase Error] In the presence of significant channel noise, the system of the equations \ref{LinearSystemofEquation} does not hold. To overcome the difficulty, we  formulate the following minimization problem:
	\begin{align}
	\boldsymbol{h}^*=&\arg\min_{\boldsymbol{h}} \left\|{\bf{G}}(\tilde{\mathbf{x}}_{\mathsf{a}}^{(\ell)})\boldsymbol{h} - {\bf{b}}(\tilde{\mathbf{x}}_{\mathsf{a}}^{(\ell)})\right\|\nonumber\\
	=& {\left( {{{\bf{G}}(\tilde{\mathbf{x}}_{\mathsf{a}}^{(\ell)})^T}{\bf{G}}(\tilde{\mathbf{x}}_{\mathsf{a}}^{(\ell)})} \right)^{ - 1}}{{\bf{G}}(\tilde{\mathbf{x}}_{\mathsf{a}}^{(\ell)})^T}{\bf{b}}(\tilde{\mathbf{x}}_{\mathsf{a}}^{(\ell)}),
	\end{align}
	which has the same structure as \eqref{Solution}. Note that the resultant $\sigma$ from \eqref{Sigma} is differently calculated depending on the choice of the SV's antenna $m$, denoted by $\sigma_m$. Averaging these values gives the accurate estimate of $\sigma$ such that $\sigma=\sum_{m=1}^M \sigma_m$. 
\end{remark}

%\emph{Error Analysis on Synchronization}

\subsubsection{Performance Analysis}
Assuming that the phase error follows an \emph{independent identically distributed} (i.i.d.) Gaussian distribution ${\boldsymbol{\varepsilon }} \sim  \mathcal{N}(0, \sigma_z^2\bf{I})$, the following proposition is provided.  
\begin{prop}[Error Covariance]\label{Prop:sync} \emph{As the number of SV's antennas $M$ becomes larger, the covariance matrix ${\rm{cov}}({\bf{\tilde x}}_{\mathsf a}^{(\ell )})$ scales with $\frac{1}{M-1}$. }
\end{prop}
\begin{proof}
	See Appendix A.
\end{proof}
\begin{remark} [Effect of SV's Number of Antennas] The SV's number of antennas $M$ represents to the spatial sampling rate on the receive aperture. Larger $M$ enables to decrease the distortion level, yielding $(M-1)$ times faster convergence.   
\end{remark}

\subsection{Image Retrivel}\label{sec:Imaging}

The synchronization gap $\sigma$ can be removed by the preceding step, facilitating to solve \ref{ProblemDefinition1} in the following. Recall the synchronized demodulation \eqref{demodulated_sync} enabling to express  ${y}_m^{\ell,k}$ as a 3D surface integral form as
\begin{align}\label{imageSignal}
{y}_m^{\ell,k} &%\nonumber \\
\!=\! \Gamma^{(\ell)}\!\!\!\int_{\mathbb{R}^3}\!\!\! {\mathsf{I}_{\mathbf{x}^{\!(\ell)}}{\exp\!\left(\!\! - j\frac{2\pi\! f_k}{c}\!\sqrt {\!(\mathbf{x}^{\!(\ell)}\!-\!\mathbf{p}_m)(\mathbf{x}^{\!(\ell)}\!-\!\mathbf{p}_m)^{\!T}} \right)}d{\mathbf{x}}^{\!(\ell)}}, %\nonumber \\ 
\end{align}
where $\sqrt {(\mathbf{x}^{(\ell)}-\mathbf{p}_m)(\mathbf{x}^{(\ell)}-\mathbf{p}_m)^T}$ represents the Euclidean distance between point $\mathbf{x}^{(\ell)}$ and the location of the SV's antenna $m$, denoted by $\mathbf{p}_m=[x_m, y_m, z_0]$.

{ 	
	%    Here we first give the principle of the imaging algorithm intuitively.
	%		\item The expression of \eqref{imageSignal} is similar to the Fourier transform of $\mathsf{I}_\mathbf{x}$ while direct Fourier transform cannot be applied here because the expression inside the exponential term is not linear.
	Let $\mathbf{f}_k=\left[f_k^{(x)}, f_k^{(y)}, f_k^{(z)}\right]$ denote the vector of which the components represent the spatial frequencies to the corresponding directions, namely,
	$f_k=\sqrt {\left(f_{k}^{(x)}\right)^2 + \left(f_{k}^{(y)}\right)^2 + \left(f_{k}^{(z)}\right)^2}$.
	The spherical wave in \eqref{imageSignal} can be decomposed into an infinite superposition of plane waves 
	by rewriting the exponential term in terms of $\mathbf{f}_k$ as 
	\begin{align}\label{Decompose1}
	{y}_m^{\ell,k} \!=\!\Gamma^{\!(\ell)}\!\!\!\!\int_{\mathbb{R}^3}\!\!\! {\mathsf{I}_{\mathbf{x}^{\!(\ell)}}\!\!\left\{\!\int_{\!\sqrt{\mathbf{f}_k \mathbf{f}_k^T}\!=\!f_k}\!\!\!\!\exp\!\left(\!\!-j\!\frac{2\pi}{c} \mathbf{f}_k (\mathbf{x}^{\!(\ell)}\!\!-\!\mathbf{p}_m)^{\!T} \!\right)\!\!d{\mathbf{f}_k} \!\right\}\!d{\mathbf{x}}^{\!(\ell)}},%\nonumber\\
	\end{align}
	{which is proved in~\cite{FourierProof} without approximation.}
	
	According to \eqref{Decompose1}, the received signals can be rewritten with Fourier and inverse Fourier transform on $\mathbf{x}^{(\ell)}$ and $\mathbf{f}_k$, respectively. With such relation between the transmit antenna location and the received data, the image $\mathsf{I}_{\mathbf{x}^{(\ell)}}$ of VP $\ell$ can be straightforward obtained from ${y}_m^{\ell,k}$, which is given as follows.
}

	\begin{align}\label{resonstruction}
	\tilde{\mathsf{I}}_{\mathbf{x}^{(\ell)}}\!=&
	\mathsf{FT}_{\mathsf{3D}}^{ - 1}\bigg\{\!\frac{\mathsf{FT}_{\mathsf{2D}}\left( \{{{y}_m^{\ell,k}}\}_{m=1}^M, f_k^{(x)}, f_k^{(y)}\right)}{{\Gamma^{(\ell)}}}\times\nonumber\\
	&\exp\!\bigg(\!\!\!-\! j\frac{2\pi}{c}\!\underbrace {\!\sqrt {\!{(f_k)^2}\!-\! \left(\!f_{k}^{(x)}\!\right)^{\!2} \!-\! \left(\!f_{k}^{(y)}\!\right)^{\!2}}}_{f_k^{(z)}} \!\!{z_0}\!\bigg), \mathbf{x}^{(\ell)}\!\bigg\},
	\end{align}
where $\mathsf{FT}_{\mathsf{3D}}^{-1}$, $\mathsf{FT}_{\mathsf{2D}}$ represent 3-D inverse Fourier transform and 2-D Fourier transforms, respectively. 
{Note that due to the finite and discrete deployment of antennas at both the TV and SV, discrete Fourier transform and inverse discrete Fourier transform are used, and the estimated $\tilde{\mathsf{I}}_{\mathbf{x}^{(\ell)}}$ is a continuous  value between $[0,1]$.} It is thus necessary to map $\tilde{\mathsf{I}}_{\mathbf{x}^{(\ell)}}$ into either one or zero, namely,
\begin{align}\label{Mapping_Rule}
{\mathsf{I}}_{\mathbf{x}^{(\ell)}}=\left\{
\begin{aligned}
& 1, && \textrm{if $|\tilde{\mathsf{I}}_{\mathbf{x}^{(\ell)}}| \geq\nu$},\\
& 0, && \textrm{otherwise},
\end{aligned}
\right. \quad \mathbf{x}^{(\ell)}\in \mathbb{R}^3
\end{align}
where $\nu$ represents the detection threshold. 

%{\zzz Imaging algorithms based on the Fourier transform is commonly used in different applications [xxSAR, ISAR, mimo]. For example, by ultilizing Fourier transforms, the imaging process in MS array imaging and MIMO array imaging systems can be presented as 
%\begin{align}\label{monoImaging}
%\tilde{\mathsf{O}}_\mathbf{x}=\mathsf{FT}_{\mathsf{3D}}^{ - 1}\left\{\mathsf{FT}_{\mathsf{2D}}\left( \{{y}_m^{(k)}\}_{m=1}^M, f_k^{(x)}, f_k^{(y)}\right)
%\exp\left(- j\frac{4\pi}{c}{f_k^{(z)}} {z_0}\right), \mathbf{x}\right\},
%\end{align}
%and 		
%\begin{align}\label{multiImaging}
%\tilde{\mathsf{O}}_\mathbf{x}=\mathsf{FT}_{\mathsf{3D}}^{ - 1}\left\{\mathsf{FT}_{\mathsf{4D}}\left( %\{{y}_{n,m}^{(k)}\}_{n=1,m=1}^{N,M}, f_{k,n}^{(x)}, f_{k,n}^{(y)}, f_{k,m}^{(x)}, f_{k,m}^{(y)}\right)
%\exp\left(- j\frac{4\pi}{c}{(f_{k,n}^{(z)}+f_{k,n}^{(z)})} {z_0}\right), \mathbf{x}\right\},
%\end{align}	
%However, mechanical switching over transmit antennas is necessary for the spatial data sampling in both \eqref{monoImaging} and \eqref{multiImaging}, which limits the applications requiring real-time response. Note that in our work, all transmit antennas send signals at different frequencies simultaneously, which removes the switching process without aggravate the signal processing intensity. This property make the imaging technique possible to be near-real-time. }

%	{\zzz The resampling in frequency domain is given as follows, as well as the resolution of the imaging algorithm, which is related to the frequency steps and the bandwidth for imaging given in subsection~\ref{SignalModel}. }

\begin{remark}[Resampling]{To calculate the inverse 3D Fourier transform in \eqref{resonstruction}, sampling on frequency domain with constant  interval is necessary. 
		However, due to the nonlinear relation of each frequency component as $\sqrt{\mathbf{f}_k \mathbf{f}_k^T}=f_k$, regular samplings on $f_k^{(x)}$ and $f_k^{(y)}$ lead to the irregular sequence on $f_k^{(z)}$ domain. It is thus required to make the sequence regular by using an interpolation, which is called a \emph{resampling}. We use a linear interpolation for the resampling whose error is marginal verified by simulation. }
\end{remark}

\subsection{Resolution Analysis}
In this subsection, we provide analysis on resolution of the multi-point position recovered by the above algorithm.  

The direct relation between    
bandwidth $\mathcal B$ and resolution $\delta$ is established as $\delta=\frac{c}{\mathcal{B}}$,
where $c$ is the light speed. Based on this relation, azimuth and range resolutions are analyzed here. 
\subsubsection{Azimuth Resolution} 
Consider frequency $f_k$. The bandwidth in $f^{(y)}$ direction, denoted by ${\cal B}_y$ is approximately
\begin{align}
{\cal B}_y \approx\mathbb{E}_k \left[\frac{2D}{{\sqrt {{{ {{R}} }^2} + {D^2}} }} f_k\right]=\frac{2D}{{\sqrt {{{ {{R}} }^2} + {D^2}} }} f_c,
\end{align}
where $f_c=\frac{f_1+f_K}{2}$, $R$ is the aperture size and $D$ is the range. 
The azimuth resolution in $Y$ direction can be straightforward obtained as
\begin{align}\label{spatialResolution}
\delta _y \approx \frac{{c\sqrt {{{ {{R}} }^2} + {D^2}} }}{{2{f_c}D}}.
\end{align}		

%\begin{figure}
%	\centering
%	\includegraphics[%height=240pt, 
%	width=180pt]{Figures/resolution1.pdf}
%	\caption{Illustration of azimuth and range resolutions.}\label{resolution}
%\end{figure}
\subsubsection{Range Resolution} 
The bandwidth in $f_z$ direction is 
$\mathcal{B}_z \approx (f_K-f_1)$, and the range resolution is
\begin{align}\label{RangeResolution}
{\delta _z} \approx \frac{{c }}{{\left( {{f_K} - {f_1}} \right)}}.
\end{align}

\begin{remark}[Sampling  Requirements] To achieve the above resolutions, there exist two kinds of sampling requirements.
	
	{\bf Spatial Sampling:} To achieve the resolution in \eqref{spatialResolution}, the receive antennas deployment needs to meet the Nyquist sampling criterion. Therefore, the distances between two adjacent receive antennas along $X$ or $Y$ directions are
	\begin{align}\label{samplingReq}
	{\Delta _x} = {\Delta _y} \le \min_R \bigg\{\frac{c}{{2{f_c}}}\frac{{\sqrt {{{ {{R}} }^2} + {D^2}} }}{D}\bigg\}\mathop  = \limits^{(a)} \frac{c}{{2{f_c}}},
	\end{align}
	where (a) follows for the worst case with $R=0$.
	
	{\bf Frequency sampling:}
	The frequency sampling interval refers to the minimum frequency gap $\Delta$ to achieve the resolution $\delta_z$ \eqref{RangeResolution}, given as $\Delta\leq \frac{c}{R_{\max}}$.
\end{remark}

		\subsection{Multi-Path Position Mapping}\label{sec2:Common-point}
		
		In this subsection, we aim at reconstructing the TV's RP  using multiple VPs
		by solving \ref{ProblemDefinition2} under the assumption that the reflection surfaces are vertical to the ground, which is common in practice. To this end, we use  the representative points derived in Sec. \ref{sec:Synchronization} whose coordinates are $\mathbf{x}^{(\ell)}_{\mathsf{a}}=\left(x^{(\ell)}_{\mathsf{a}},y^{(\ell)}_{\mathsf{a}},z^{(\ell)}_{\mathsf{a}}\right)$ and $\mathbf{x}^{(\ell)}_{\mathsf{b}}=\left(x_{\mathsf{b}}^{(\ell)},y_b^{(\ell)},z_{\mathsf{b}}^{(\ell)}\right)$, which have geographical relations with the counterpart points on the RP denoted by $\mathbf{x}_{\mathsf{a}}=\left(x_{\mathsf{a}},y_{\mathsf{a}},z_{\mathsf{a}}\right)$ and $\mathbf{x}_{\mathsf{b}}=\left(x_{\mathsf{b}},y_{\mathsf{b}},z_{\mathsf{b}}\right)$  summarized in Lemma \ref{CommonPointRelation}. Using the properties,
		the feasible condition of RP reconstruction is derived and the corresponding algorithm is designed based on the feasible condition. 
		
		%	For VP $\ell$, the SV distinguishes the corresponding two common points $\mathbf{x}_{\mathsf{a}}^{(\ell)}=(x_{\mathsf{a}}^{(\ell)},y_{\mathsf{a}}^{(\ell)},z_{\mathsf{a}}^{(\ell)})$ and $\mathbf{x}_{\mathsf{b}}^{(\ell)}=(x_{\mathsf{b}}^{(\ell)},y_{\mathsf{b}}^{(\ell)},z_{\mathsf{b}}^{(\ell)})$, which can be accomplished in the synchronization procedure (See Sec. XX). These detected common points help the SV to combine  multiple VPs into one real image.  
		
		%	\begin{figure}
		%		\centering
		%		\includegraphics[height=200pt, width=260pt]{Figures/multimirror.pdf}
		%		\caption{Reconstruction of multiple virtual images assisted by multiple MVs.}\label{multimirror}
		%	\end{figure}
		
		%	In this paper, we only consider the condition where the MVs' sides are vertical to the ground, which satisfies most cases in practice. Without this assumption, the problem can also be solved similarly, while the solution is more complicated and more common points are needed in each VP. The approach to fix the location of the TV is illustrated below.
		
		%	Firstly, three couples of common points $\{(\mathbf{n}_{\mathsf{a}}^{(\ell)},\mathbf{n}_{\mathsf{b}}^{(\ell)})\}$,  are obtained from three VPs in the synchronization step, and each couple of points forms a line segment. The lengths of these line segments should be the same. 
		
		\begin{lemma}\label{CommonPointRelation} {Consider VPs $\ell_1$ and $\ell_2$ whose representative points are $\{\mathbf{x}_{\mathsf{a}}^{(\ell_1)}, \mathbf{x}_{\mathsf{b}}^{(\ell_1)}\}$ and $\{\mathbf{x}_{\mathsf{a}}^{(\ell_2)}, \mathbf{x}_{\mathsf{b}}^{(\ell_2)}\}$ respectively, 
				which have relations with the counterpart points of the real position denoted by $\left(\mathbf{x}_{\mathsf{a}},\mathbf{x}_{\mathsf{b}}\right)$ as follows.
				\begin{enumerate}
					\item Let $\theta_{\ell}$ denote the directed angle from the $X$-axis (SV's moving direction) to the virtual line between $\mathbf{x}_{\mathsf{a}}^{(\ell)}$ and $\mathbf{x}_{\mathsf{a}}$ or $\mathbf{x}_{\mathsf{b}}^{(\ell)}$ and $\mathbf{x}_{\mathsf{b}}$ (see Fig. \ref{location}). The relation between $\mathbf{x}_{\mathsf{a}}$ is given in terms of $\theta_{\ell_1}$ and $\theta_{\ell_2}$ as
					\begin{align}\label{representXA}
					\left(\!\begin{aligned}
					& x_{\mathsf{a}}\\
					& y_{\mathsf{a}}\\
					& z_{\mathsf{a}}
					\end{aligned}\!\right)\!\!=\!\!\left(\!
					\begin{aligned}
					&	\frac{{\!\left(\! {{z_{\mathsf{a}}^{\!(\ell_1)}} \!\!-\!\! {z_{\mathsf{a}}^{\!(\ell_2)}}} \!\right) \!\!+\!\! \left(\! {{x_{\mathsf{a}}^{\!(\ell_2)}}\!\tan\! {(\theta_{\ell_2})} \!-\! {x_{\mathsf{a}}^{\!(\ell_1)}}\!\tan\! {(\theta_{\ell_1})}}\! \right)}}{{\theta_{\ell_2} \!-\! \theta_{\ell_1}}} \\
					&      \textrm{$y_{\mathsf{a}}^{(\ell_1)}$ or $y_{\mathsf{a}}^{(\ell_2)}$}  \\
					&	\textrm{${z_{\mathsf{a}}^{(\ell_1)}} + \tan \left({\theta_{\ell_1}}\right)\left( {x_{\mathsf{a}} - {x_{\mathsf{a}}^{(\ell_1)}}} \!\right)$}
					\end{aligned}
					\right), 
					\end{align}
					and the relation between $\mathbf{x}_{\mathsf{b}}$ is given in terms of $\theta_{\ell_1}$, and $\theta_{\ell_2}$ is obtaining by replacing all $\mathsf{a}$ in \eqref{representXA} with $\mathsf{b}$. 
					\item Let $\phi_{\ell}$ denote the directed angle from the $X$-axis to the line segment of VP $\ell$ as shown in Fig. \ref{location}. The angles $\theta_{\ell}$ and $\phi_{\ell}$ enable to make a relation between two VPs $\ell_1$ and $\ell_2$ as
					\begin{align}\label{angleRelation}
					\theta_{\ell_1}-\theta_{\ell_2}=\frac{\phi_{\ell_1}-\phi_{\ell_2}}{2}. 
					\end{align}
				\end{enumerate}
			}
		\end{lemma}
		\begin{proof}
			See Appendix B. 
		\end{proof}
		
		Some intuitions are made from Lemma \ref{CommonPointRelation}. First, it is shown in \eqref{representXA} that the coordinates of  $\left(\mathbf{x}_{\mathsf{a}},\mathbf{x}_{\mathsf{b}}\right)$ can be calculated when $\theta_{\ell_1}$ and $\theta_{\ell_2}$ are given. Second, noting that $\phi_{\ell_1}$ and $\phi_{\ell_2}$ are observable from VPs $\ell_1$ and $\ell_2$ directly, 
		$\theta_{\ell_2}$ is easily calculated when $\theta_{\ell_1}$ is given. Last, $\theta_{\ell_1}$ and the resultant $\left(\mathbf{x}_{\mathsf{a}},\mathbf{x}_{\mathsf{b}}\right)$ are said to be correct if another combination of two VPs, for example $\ell_1$ and $\ell_3$, can yield the equivalent result of $\left(\mathbf{x}_{\mathsf{a}},\mathbf{x}_{\mathsf{b}}\right)$. 
		As a result, we can lead to the following feasibility condition. 
		%\begin{prop}[Feasibility Condition of NLoS Position Mapping]\label{Prop2}\emph{To detect the real multi-point position of TV, at least three VPs are required: $L\geq3$.}
		%\end{prop}
		%\begin{proof}
		%	See Appendix C. 
		%\end{proof}
		
		\begin{figure}
			\centering
			\includegraphics[scale = 0.35]{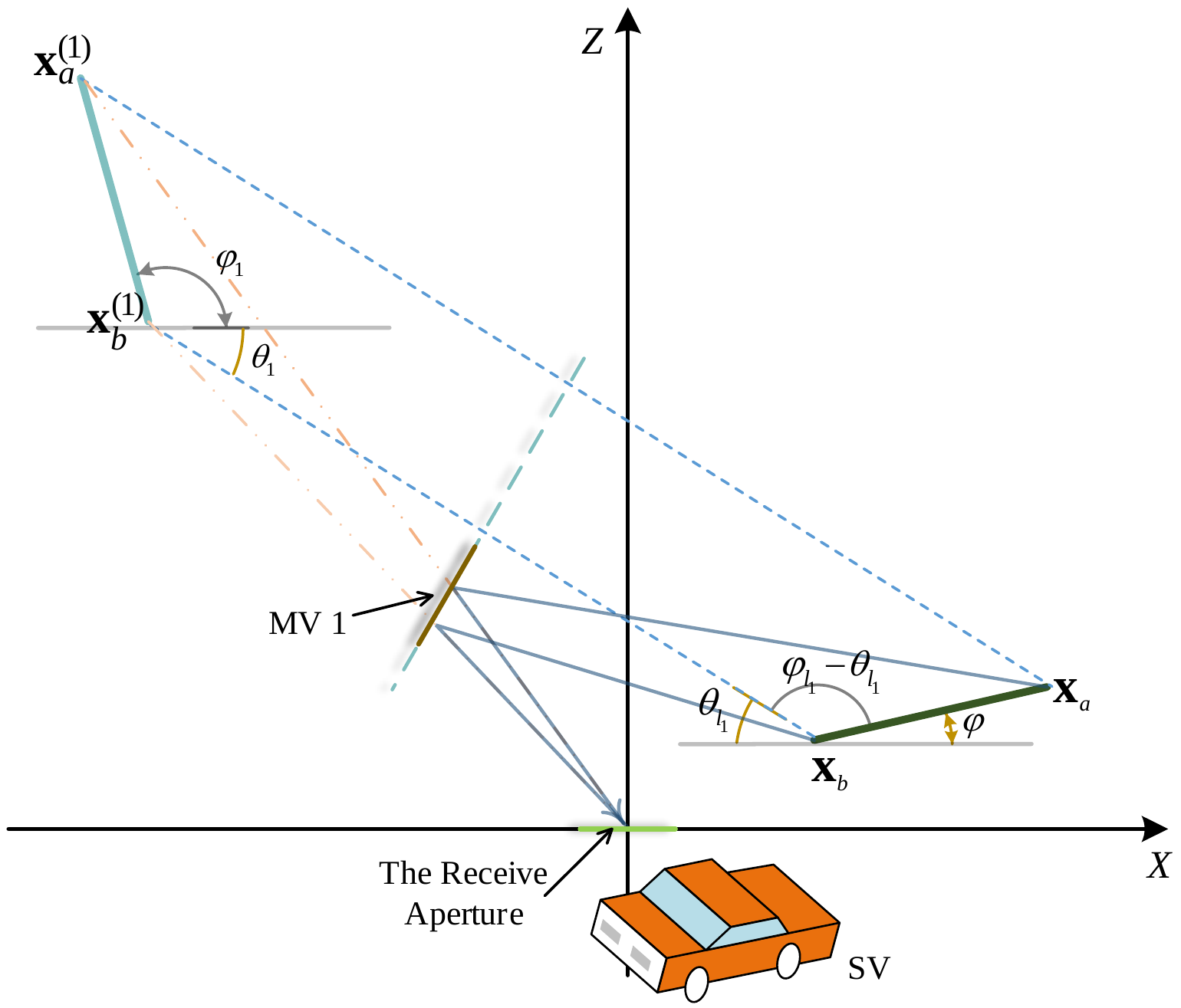}
			\caption{Relations between VPs and the TV.}\label{location}
		\end{figure}
		
		Based on Lemma \ref{CommonPointRelation} and Proposition \ref{Prop2}, the representative points $\left(\mathbf{x}_{\mathsf{a}},\mathbf{x}_{\mathsf{b}}\right)$ are estimated as follows. 
		Consider the angle $\theta_{1}$ is given. According to \eqref{angleRelation}, the other angles $\{\theta_{\ell}\}_{\ell=2}^L$ is then expressed in terms of $\theta_{1}$ as $\theta_{\ell}=\theta_{1}+\frac{\phi_{\ell}-\phi_{1}}{2}$. By plugging one pair of $(\theta_{1}, \theta_{\ell})$ into \eqref{representXA}, we can calculate the corresponding location of two representative points $\mathsf{a}$ and $\mathsf{b}$, denoted by $({\mathbf{z}}_{\mathsf{a}}^{(\ell)}, {\mathbf{z}}_{\mathsf{b}}^{(\ell)})$, which is equal to $\left(\mathbf{x}_{\mathsf{a}},\mathbf{x}_{\mathsf{b}}\right)$ when the given angle $\theta_{1}$ is correct. In other words, estimating $\left(\mathbf{x}_{\mathsf{a}},\mathbf{x}_{\mathsf{b}}\right)$ is translated into finding $\theta_{1}$ minimizing the following squared Euclidean distance as
		\begin{align}\label{OptTheta}\tag{E5}
		\theta_{1}^*=\arg\min_{\theta_{1}}\sum\limits_{p=2}^{L}{\sum\limits_{q=2}^{L} {\left(\left\| {\mathbf{z}}_{\mathsf{a}}^{(q)} - {\mathbf{z}}_{\mathsf{a}}^{(p)} \right\|+\left\| {\mathbf{z}}_{\mathsf{b}}^{(q)} - {\mathbf{z}}_{\mathsf{b}}^{(p)} \right\|\right) }}
		\end{align}
		where ${\mathbf{x}}_{\mathsf{a}}^*=\frac{1}{L-1}\sum_{\ell=2}^L{\mathbf{z}}_{\mathsf{a}}^{(\ell)}$ and 
		${\mathbf{x}}_{\mathsf{b}}^*=\frac{1}{L-1}\sum_{\ell=2}^L{\mathbf{z}}_{\mathsf{b}}^{(\ell)}$. The optimal $\theta_{1}^*$ is estimated by 1D search over $ \left[ { - {\pi },{\pi }} \right]$, and the resultant $({\mathbf{x}}_{\mathsf{a}}^*, {\mathbf{x}}_{\mathsf{b}}^*)$ corresponds to the optimal $\left(\mathbf{x}_{\mathsf{a}},\mathbf{x}_{\mathsf{b}}\right)$. Note that in case without phase error, the optimal solution of the problem \ref{OptTheta} is zero and $({\mathbf{x}}_{\mathsf{a}}^*, {\mathbf{x}}_{\mathsf{b}}^*)=\left(\mathbf{x}_{\mathsf{a}},\mathbf{x}_{\mathsf{b}}\right)$.

		\begin{remark}[Existence of LoS path]{The LoS case is a special realization of NLoS case, where one couple of representative points $\left(\mathbf{x}_{\mathsf{a}}^{(0)},\mathbf{x}_{\mathsf{b}}^{(0)}\right)$ are equivalent to the exact location $\left(\mathbf{x}_{\mathsf{a}},\mathbf{x}_{\mathsf{b}}\right)$. Therefore, all mathematical expression for position mapping still holds in the LoS condition, and the resultant $\left(\mathbf{x}_{\mathsf{a}}^*,\mathbf{x}_{\mathsf{b}}^*\right)$ can be obtained in the same way. }	
		\end{remark}

		From the estimated $\theta_{\ell}^{*}$ and ${\mathbf{x}}_{\mathsf{a}}^*$ (or ${\mathbf{x}}_{\mathsf{b}}^*$), the real position $\{\mathsf{I}_{\mathbf{x}}\}$ can be obtained by shifting the MI $\ell$, $\{\mathsf{I}_{\mathbf{x}^{(\ell)}}\}$. The reflection plane $\ell$, which is located on the middle of ${\mathbf{x}}_{\mathsf{a}}^*$ and ${\mathbf{x}}_{\mathsf{a}^{(\ell)}}$, can be expressed by a line because it is perpendicular to $X-Y$ plane such that
		\begin{align}\label{reflectionPlane}
		z = - \frac{1}{{\tan ({\theta_{\ell}^*})}}\left(x- \frac{{x_{\mathsf{a}}^{(\ell)}} + {x_{\mathsf{a}}^*}}{2} \right)+\frac{\left( {{z_{\mathsf{a}}^{(\ell)}}+{z_{\mathsf{a}}^*}} \right)}{2},
		\end{align}
		where $\mathbf{x}_{\mathsf{a}}^*=({x}_{\mathsf{a}}^*, {y}_{\mathsf{a}}^*, {z}_{\mathsf{a}}^*)$. 
		It is worth noting that the  line \eqref{reflectionPlane} passes the middle points of all lines between the real points $\mathbf{x}$ and virtual points $\mathbf{x}^{(\ell)}$, leading to the following mapping rule. 
		\begin{prop}[Position Mapping]\label{Prop3} \emph{Consider the VP $\ell$ represented by $\{\mathsf{I}_{\mathbf{x}^{(\ell)}}\}$. Given $\theta_{\ell}$ and ${\mathbf{x}}_{\mathsf{a}}^*$ (or ${\mathbf{x}}_{\mathsf{b}}^*$), the RP of the TV $\{\mathsf{I}_{\mathbf{x}^{(\ell)}}\}$ can be obtained by the following mapping function:
				\begin{align}\label{FindLocation}
				\mathsf{I}_{\mathbf{x}}=\mathsf{I}_{\mathsf{G}(\mathbf{x}^{(\ell)})},\quad\mathsf{G}(\mathbf{x}^{(\ell)})=\left(\begin{aligned}
				& x^*\\
				& y^*\\
				& z^*
				\end{aligned}\right)
				\end{align}
				where
				\begin{align}
				\left\{\!\!\! \begin{array}{l}
				x^*\!=\!{x^{(\ell )}} \!+\!\frac{\left( {1 + {{\tan }^2}(\theta _\ell ^*)} \right)}{{\tan (\theta _\ell ^*)}}\!\left(\! {\frac{{x_{\mathsf{a}}^{(\ell )} + x_{\mathsf{a}}^*}}{{\tan (\theta _\ell ^*)}} \!+\! z_{\mathsf{a}}^{(\ell )} \!\!+\! z_{\mathsf{a}}^* \!-\! \frac{{2{x^{(\ell )}}}}{{\tan (\theta _\ell ^*)}} \!-\! 2{z^{(\ell )}}} \right)\\
				y^*\!=\!y^{(\ell)}\\
				z^*\!=\!{z^{(\ell )}} \!+\! \left( {1 \!+\! {{\tan }^2}(\theta _\ell ^*)} \right)\!\left(\! {\frac{{x_{\mathsf{a}}^{(\ell )} + x_{\mathsf{a}}^*-{2{x^{(\ell )}}}}}{{\tan (\theta _\ell ^*)}} \!+\! z_{\mathsf{a}}^{(\ell )} \!+\! z_{\mathsf{a}}^* \!-\! 2{z^{(\ell )}}} \!\right)\end{array}. \right.\nonumber
				\end{align}
			}
		\end{prop}
		\begin{proof}
			See Appendix D. 
		\end{proof}

		%\section{Performance Analysis}	
		%	In this section, we give analytical performance evaluation on operations in Sec. \ref{sec:Synchronization} and Sec. \ref{subsce: Compare}, while the performance of the algorithm in Sec. \ref{sec2:Common-point} is provided in the next section by simulations. The performance metric for synchronization is the covariance matrix of the location estimates. The analysis on the resolution is also provided for the imaging algorithm. Moreover, the system setting requirements of the imaging algorithm and the comparison to existing mmWave imaging techniques are also given.

		\section{Simulation Results}
		
		In this section, the performance of the proposed multi-point vehicular positioning techniques are evaluated by realistic settings. 
		SW at $2$ frequencies are used for the synchronization procedure. 
		The number of frequencies used in the SFCW $\boldsymbol{s}(t)$ \eqref{transmitSig} is $K=512$ in $57\sim60$ GHz with the constant gap $\Delta = 5.86$ MHz. The two frequencies used in the SWs $\boldsymbol{s}_\mathsf{a}(t)$ are $(57-\Delta \cdot i)$ GHz where $i=1,2$. 
		The numbers of the TV's  antennas $N$ is $200$ and the number of SV's  antennas $M$ is $256$, uniformly deployed on the receive aperture with size $1 \times 1$ $m^2$. SNR of each received signal is fixed to $10$ dB.
		For the performance metric, we use the Hausdorff distances defined as follows. 
		\begin{definition}[Hausdorff distance~\cite{Hausdorff}] Consider two images $\mathcal{A}$ and $\mathcal{B}$. The Hausdorff distance is defined as 
			\begin{align}\label{Hausdorff_Def}
			\mathsf{H}\left( {\mathcal{A},\mathcal{B}} \right) = \max \left( {\mathsf{h}\left( {\mathcal{A},\mathcal{B}} \right), \mathsf{h}\left( {\mathcal{B},\mathcal{A}} \right)} \right),
			\end{align} 
			where $\mathsf{h}\left( {\mathcal{A},\mathcal{B}} \right) = \mathop {\max }\limits_{a \in \mathcal{A}} \mathop {\min }\limits_{b \in \mathcal{B}} \left\| {a - b} \right\|$ is the direct Hausdorff distance from $\mathcal{B}$ to $\mathcal{A}$. 
		\end{definition}

		\begin{figure}
			\centering
			\includegraphics[scale = 0.22]{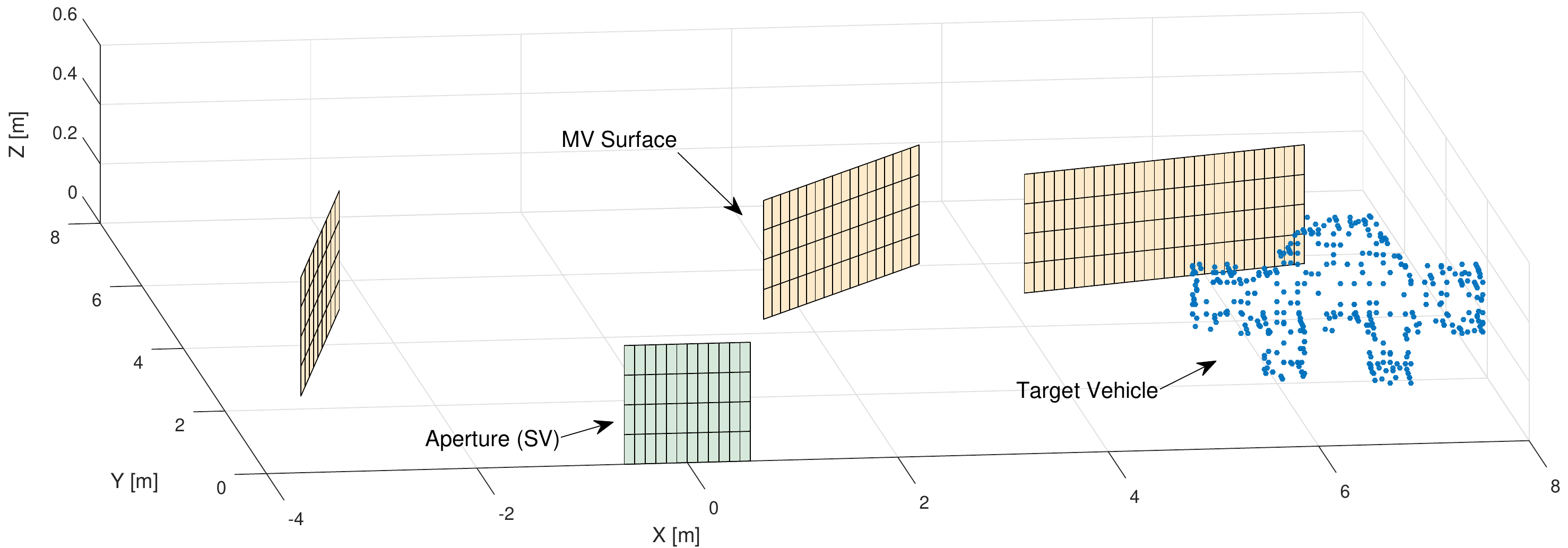}
			\caption{Initial topology.}\label{InitialSetting}
		\end{figure}
		\begin{figure}
			\centering
			\includegraphics[scale = 0.23]{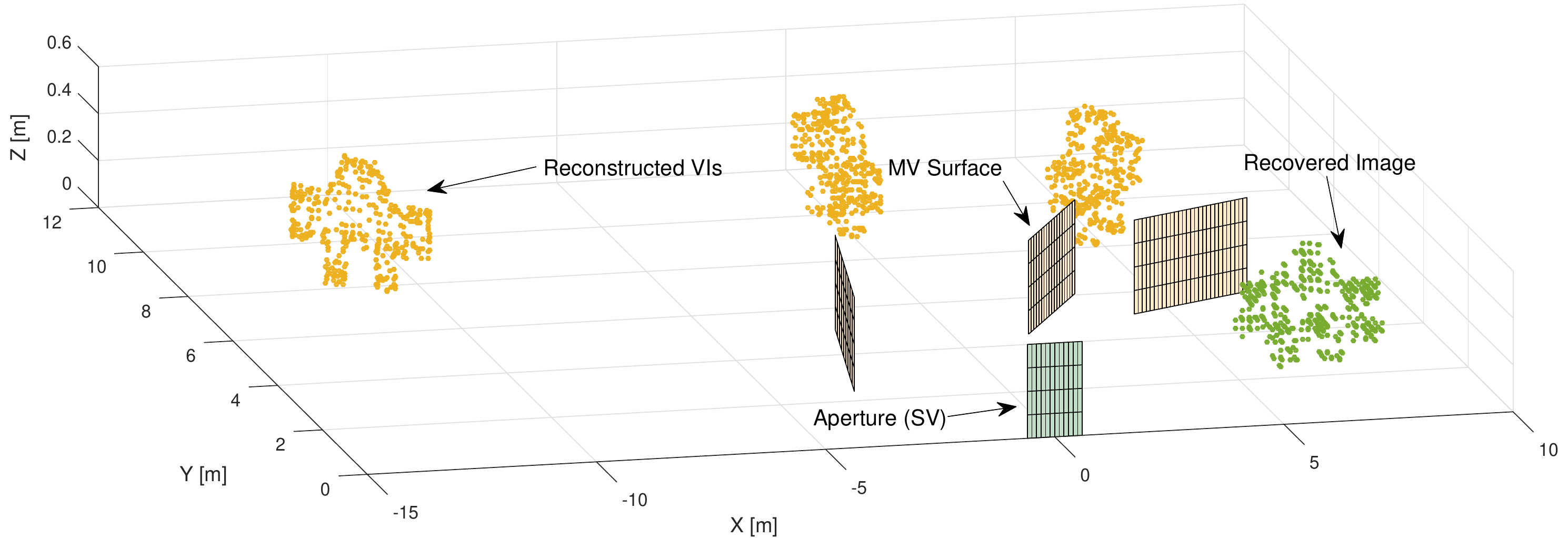}
			\caption{The detection of the VPs.}\label{Reconstruct1st}
		\end{figure}
		
		\subsection{Graphical Example of Multi-Point Positioning}
		
		This subsection aims at explaining the entire vehicular positioning procedure with graphical examples. Fig. \ref{InitialSetting}(a) shows the original shape of the TV with the size $3\times1\times0.6$ $m^3$. We consider the topology with three MVs illustrated in Fig. \ref{InitialSetting}(b). 
		%The initial topology is illustrated in Fig.~\ref{InitialSetting}.  
		The $X$, $Y$, $Z$-axes are the SV's moving direction, height and depth, respectively.
		The TV is represented by discrete points, each of which is  one TV's  antenna.
		The receive aperture is parallel to $Y$-axis and located on $(0, 0, 0)$.
		The equations of three MVs are $z = 1.02x + 3$, $z = \frac{{x + 13}}{4}$, and $z = 3x + 4$.

		Using the MVs, the SV detects three VPs represented by yellow slots in Fig. \ref{Reconstruct1st}, each of which is  differentiable using AoA information. By the mapping algorithm in Proposition \ref{Prop3}, each VP can be shifted to its real location represented by green spots with Hausdorff distance and direct Hausdorff distance $0.355$m and $0.143$ m respectively, relatively small compared to the size of the TV. The final detected position is obtained as in Fig. \ref{EnvelopCompare}(b), similar to the original one in Fig. \ref{EnvelopCompare}(a).

		% The Hausdorff distance between the reconstructed image $\mathcal{A}$ and the ideal TV model $\mathcal{B}$ is $\mathsf{H}(A,B)=0.1262m$, which is relatively small compared to the size of the TV model.

		%	Without loss of generality, signals received from the direction of the 1-st VP is selected. Therefore, the SV can receive the signals reflected by the 1-st MV and distinguish the signature waveforms from the common points in the 1-st VP. Then the synchronization procedures are taken and the location of the common points can be figured out. After synchronization, the SV can reconstruct the 1-st VP based on the received signals. The performance of the VPs is shown in Fig.~\ref{Reconstruct1st}. The reflection surface and the receive aperture in SV are also plotted.
		
		\begin{figure}[!htpb]
			\begin{minipage}[t]{0.45\linewidth}%设定图片下字的宽度，在此基础尽量满足图片的长宽
				\centering
				\includegraphics[scale = 0.25]{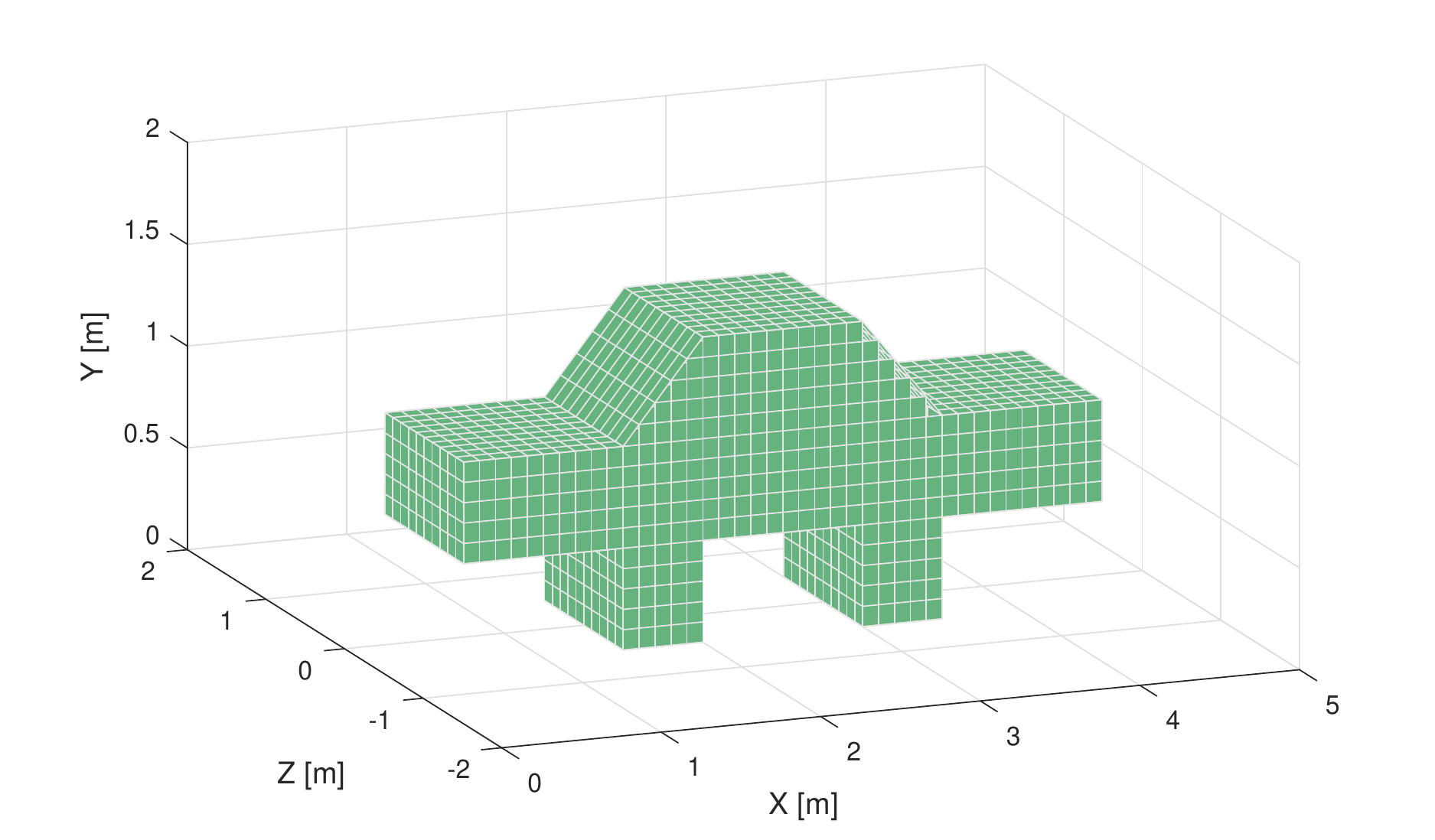}
				\caption*{(a) Envelop Diagram of TV.}%加*可以去掉默认前缀，作为图片单独的说明
				\label{fig:side:a}
			\end{minipage}
			\begin{minipage}[t]{0.45\linewidth}%需要几张添加即可，注意设定合适的linewidth
				\centering
				\includegraphics[scale = 0.33]{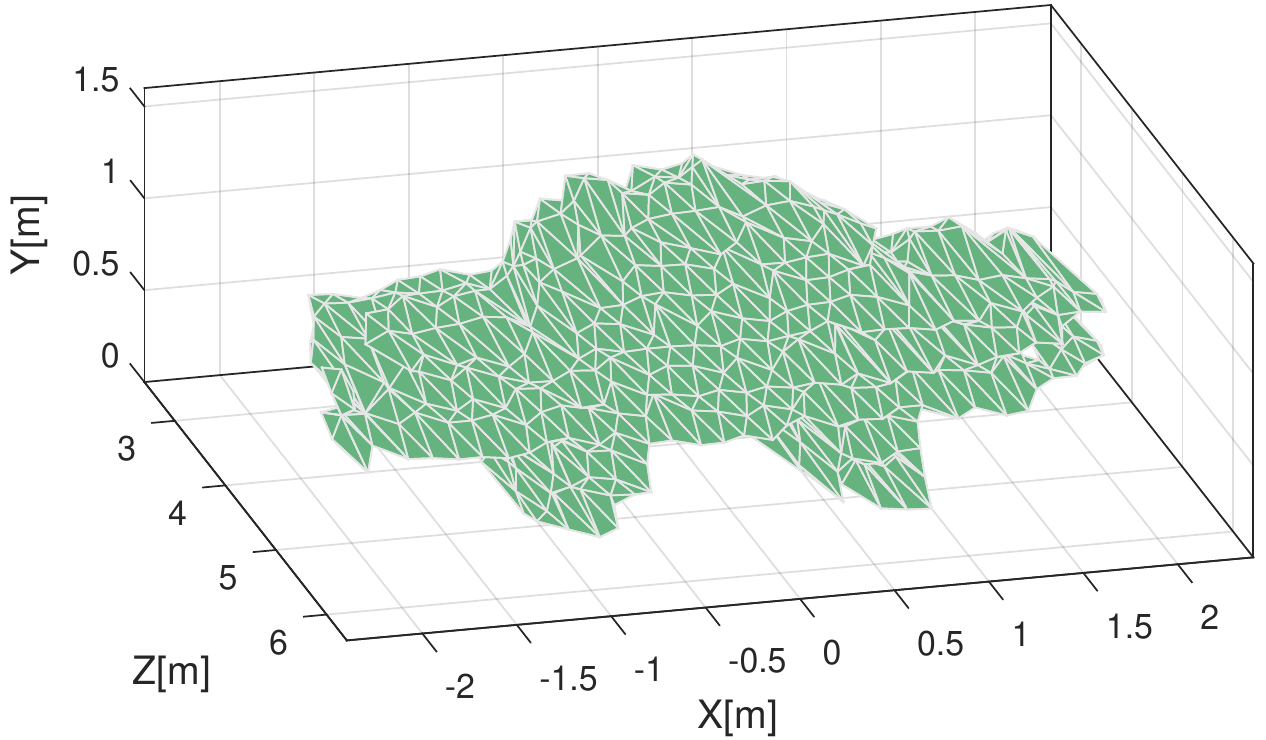}
				\caption*{(b) Envelop Diagram of the reconstructed RP.}
				\label{fig:side:b}
			\end{minipage}
			\caption{The envelop digaram comparison of the TV model and the detected position.}\label{EnvelopCompare}
		\end{figure}
		%	\begin{figure}
		%		\centering
		%		\includegraphics[scale = 0.6]{Figures/Hausdorff.pdf}
		%		\caption{Hausdorff Distance of the recovered image versus TV-SV distance.}\label{Hausdorff}
		%	\end{figure}
		%	\begin{figure}
		%		\begin{minipage}[t]{0.56\linewidth}%设定图片下字的宽度，在此基础尽量满足图片的长宽
		%			\centering
		%			\includegraphics[scale = 0.6]{Figures/DisErrorMirror.pdf}
		%			\caption*{(a) Mean Location Error versus the TV-SV distance.}%加*可以去掉默认前缀，作为图片单独的说明
		%			\label{Distance}
		%		\end{minipage}
		%		\begin{minipage}[t]{0.56\linewidth}%需要几张添加即可，注意设定合适的linewidth
		%			\centering
		%			\includegraphics[scale = 0.6]{Figures/NumMVError.pdf}
		%			\caption*{(b) Mean Location Error versus the number of MVs where the distance is $8$m {\zzz(Old simulation results)}.}
		%			\label{MVNumber}
		%		\end{minipage}
		%		\caption{The performance of the reconstructed images.}\label{PerformanceCompare}
		%	\end{figure}
		
		\subsection{Effect of Distance between SV and TV}
		In Fig.~\ref{HausdorffDis}, the Hausdorff distances are given under different TV-SV distances, showing that the positioning quality is degraded as the distance increases.  This phenomenon can be explained by \eqref{spatialResolution}, where the spatial resolution becomes poor when the detection range $R$ is large. Thus the SV may not be able to capture the clear position of the TVs far away, especially when the number of MVs is small. Moreover, the relation between the Hausdorff distance and the TV-SV distance is nonlinear. With noise in consideration, the Hausdorff distances increase faster when the TV-SV distance becomes larger.
		\begin{figure}[!htpb]
			\centering
			\includegraphics[%height=240pt, 
			width=180pt]{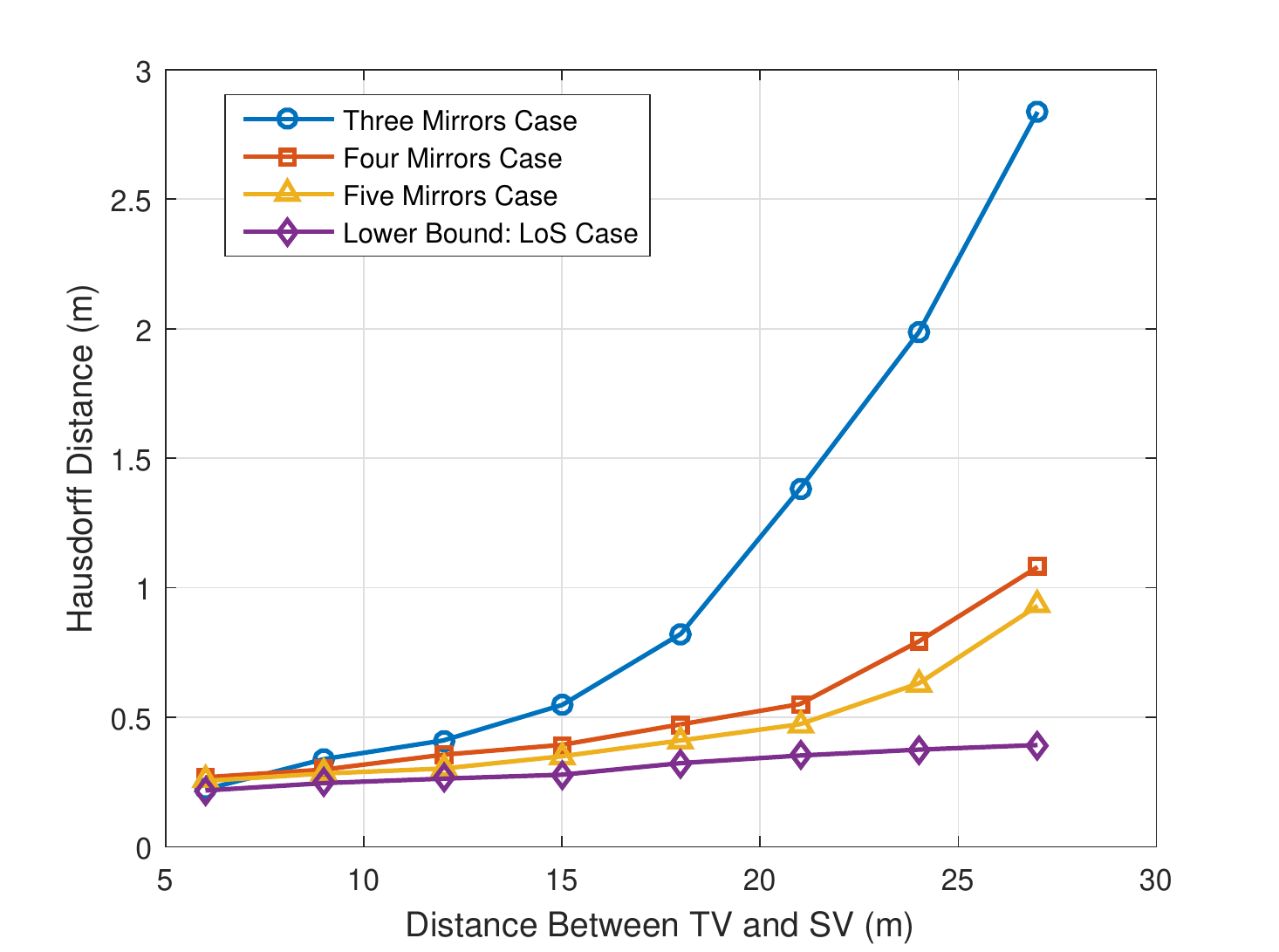}
			\caption{The performance of multi-point poditioning versus distance.}\label{HausdorffDis}
		\end{figure}
		\subsection{Effect of Mirror Vehicle Number}
		The relation between the performance and the number of MVs is also presented in Fig.~\ref{HausdorffDis}. Signals reflected from any three MVs give one estimation of $\theta_1$, but the resultant positioning quality is low due to phase error. Larger $L$ provides more combinations to estimate $\theta_1$, providing more accurate estimation of $\theta_1$ by canceling out individual estimated error. 
		The LoS case is considered as a lower bound because the operation in Sec. \ref{sec2:Common-point} is not involved.

		\section{Conclusion}
		In this paper, a novel multi-point vehicular positioning approach via mmWave signal transmissions has been proposed to capture the shape and location information of the TVs in both LoS and NLoS. The synchronization issue has been well addressed by a novel PDoA-based hyperbolic positioning approach. The predictable reflections of mmWave frequency on the surface of vehicles have been exploited to form the geometry relation for obtaining the RP from multiple VPs.
		Existing vehicular positioning techniques relies on the LoS path while the proposed techniques overcome such limitations and opens a new area of mmWave-based vehicular positioning. We believe the proposed technique enables more intelligent and safer autonomous driving and the potential of mmWave-based sensing can still be activated in the future.

		\appendix
		\subsection{Proof of proposition \ref{Prop:sync}}\label{A}
		Due to the assumption of Gaussian phase error, the covariance of the location estimation can be expressed as
		\begin{align}\label{Ap_Cov}
		{\rm{cov}}({\bf{\tilde x}}_{\mathsf a}) 
		=& \bf X {\sigma _z^2{{\bf{I}}_3}}, 
		\end{align}
		where $\bf X$ is a $3$-by-$3$ matrix as 
		\begin{align}\label{AP_EXP}
		\bf X\!&= {\left( {{\bf{G}}{{({\bf{\tilde x}}_{\mathsf {a}})}^T}{\bf{G}}({\bf{\tilde x}}_{\mathsf {a}})} \right)^{ - 1}} = \left(\sum_{m=2}^M T_{\mathsf m}({\bf{\tilde x}}_{\mathsf {a}})\right)^{ \!\!\!- 1}\nonumber\\
		&=\!\!
		\left(\!\!\sum_{m=2}^M\!\! \left[\!\!\!\! {\begin{array}{*{20}{c}}
			{{{{\left( {\frac{{\partial {\mathsf{F_m}}({\bf{\tilde x}}_{\mathsf {a}})}}{{\partial x_{\mathsf {a}}}}} \right)}^2}} }\!\!&\!\!{ {\frac{{\partial {\mathsf{F_m}}({\bf{\tilde x}}_{\mathsf {a}})}}{{\partial x_{\mathsf {a}}}}\frac{{\partial {\mathsf{F_m}}({\bf{\tilde x}}_{\mathsf {a}})}}{{\partial y_{\mathsf {a}}}}} }\!\!&\!\!{ {\frac{{\partial {\mathsf{F_m}}({\bf{\tilde x}}_{\mathsf {a}})}}{{\partial x_{\mathsf {a}}}}\frac{{\partial {\mathsf{F_m}}({\bf{\tilde x}}_{\mathsf {a}})}}{{\partial z_{\mathsf {a}}}}} }\\
			{ {\frac{{\partial {\mathsf{F_m}}({\bf{\tilde x}}_{\mathsf {a}})}}{{\partial x_{\mathsf {a}}}}\frac{{\partial {\mathsf{F_m}}({\bf{\tilde x}}_{\mathsf {a}})}}{{\partial y_{\mathsf {a}}}}} }\!\!&\!\!{ {{{\left( {\frac{{\partial {\mathsf{F_m}}({\bf{\tilde x}}_{\mathsf {a}})}}{{\partial y_{\mathsf {a}}}}} \right)}^2}} }\!\!&\!\!{ {\frac{{\partial {\mathsf{F_m}}({\bf{\tilde x}}_{\mathsf {a}})}}{{\partial y_{\mathsf {a}}}}\frac{{\partial {\mathsf{F_m}}({\bf{\tilde x}}_{\mathsf {a}})}}{{\partial z_{\mathsf {a}}}}} }\\
			{ {\frac{{\partial {\mathsf{F_m}}({\bf{\tilde x}}_{\mathsf {a}})}}{{\partial x_{\mathsf {a}}}}\frac{{\partial {\mathsf{F_m}}({\bf{\tilde x}}_{\mathsf {a}})}}{{\partial z_{\mathsf {a}}}}} }\!\!&\!\!{ {\frac{{\partial {\mathsf{F_m}}({\bf{\tilde x}}_{\mathsf {a}})}}{{\partial y_{\mathsf {a}}}}\frac{{\partial {\mathsf{F_m}}({\bf{\tilde x}}_{\mathsf {a}})}}{{\partial z_{\mathsf {a}}}}} }\!\!&\!\!{ {{{\left( {\frac{{\partial {\mathsf{F_m}}({\bf{\tilde x}}_{\mathsf {a}})}}{{\partial z_{\mathsf {a}}}}} \right)}^2}} }
			\end{array}} \!\!\!\right]\!\!\right)^{ \!\!\!- 1}\nonumber\\
		&=\frac{1}{M-1}\left(\frac{1}{M-1}\sum_{m=2}^M T_{\mathsf m}({\bf{\tilde x}}_{\mathsf {a}})\right)^{ - 1},
		\end{align}
		where the inverse matrix's each component converges to its expectation as $M$ becomes large, which becomes independent to $M$. In other words, $\bf X$ is inversely proportional to $(M-1)$, completing the proof. 
		%	Hence the values of functions $g_{xx}(\cdot)$, $g_{xy}(\cdot)$, etc. at any sampling point ${\bf p}_m$ over the aperture are continuous and bounded (by $4$), and the expectations and vairances exist. Moreover, the sampling points are i.i.d. over the aperture. Therefore, according to the law of large number, the error covariance can be further simplified as \eqref{AP_EXP} when (c) the number of samplings are sufficiently large, e.g., $32$ or $64$. Note that the variance of the matrix inversion decreases when the variance of each term in the matrix decreases.

		\subsection{Proof of Lemma \ref{CommonPointRelation} } \label{B}
		In Fig.~\ref{AngRelation}, we have $\theta^{(\ell_1)}  \in \left[ { - {\pi },{\pi }} \right]$ and the equations of line $l_{{\bf x}_{\mathsf{a}}^{(\ell_1)}{\bf x}_{\mathsf{a}}}$ and $l_{{\bf x}_{\mathsf{a}}^{(\ell_2)}{\bf x}_{\mathsf{a}}}$ are
		\begin{align}
		&l_{{\bf x}_{\mathsf{a}}^{(\ell_1)}{\bf x}_{\mathsf{a}}}:\quad{\rm{  }}z = {z_{\mathsf{a}}^{(\ell_1)}} + \tan {\theta ^{(\ell_1)}}\left( {x - {x_{\mathsf{a}}^{(\ell_1)}}} \right),\nonumber\\
		&l_{{\bf x}_{\mathsf{a}}^{(\ell_2)}{\bf x}_{\mathsf{a}}}:\quad{\rm{  }}z = {z_{\mathsf{a}}^{(\ell_2)}} + \tan {\theta ^{(\ell_2)}}\left( {x - {x_{\mathsf{a}}^{(\ell_2)}}} \right).\nonumber
		\end{align}
		Thus we have 
		\begin{align}\label{ApprepresentXA}
		{x_{\mathsf{a}}} = \frac{{\left( {{z_{\mathsf{a}}^{(\ell_1)}} - {z_{\mathsf{a}}^{(\ell_2)}}} \right) + \left( {{x_{\mathsf{a}}^{(\ell_2)}}\tan {\theta ^{(\ell_2)}} - {x_{\mathsf{a}}^{(\ell_1)}}\tan {\theta ^{(\ell_1)}}} \right)}}{{\tan {\theta ^{(\ell_2)}} - \tan {\theta ^{(\ell_1)}}}},
		\end{align}
		and
		\begin{align}
		z_{\mathsf{a}} = {z_{\mathsf{a}}^{(\ell_1)}} + \tan {\theta^{(\ell_1)}}\left( {x_{\mathsf{a}} - {x_{\mathsf{a}}^{(\ell_1)}}} \right).
		\end{align}
		The general relations in \eqref{representXA} can be derived similarly.
		Due to the symmetric relation between the VPs and the TV,  we have
		\begin{align}\label{AppAngRelation}
		\pi-\varphi {\rm{ = }}{\varphi ^{(\ell_1)}} - 2{\theta ^{(\ell_1)}} = {\varphi ^{(\ell_2)}} - 2{\theta ^{(\ell_2)}},
		\end{align}
		which can be observed from Fig.~\ref{AngRelation}. Thus \eqref{angleRelation} is obtained. 
		\begin{figure}
			\centering
			\includegraphics[scale = 0.4]{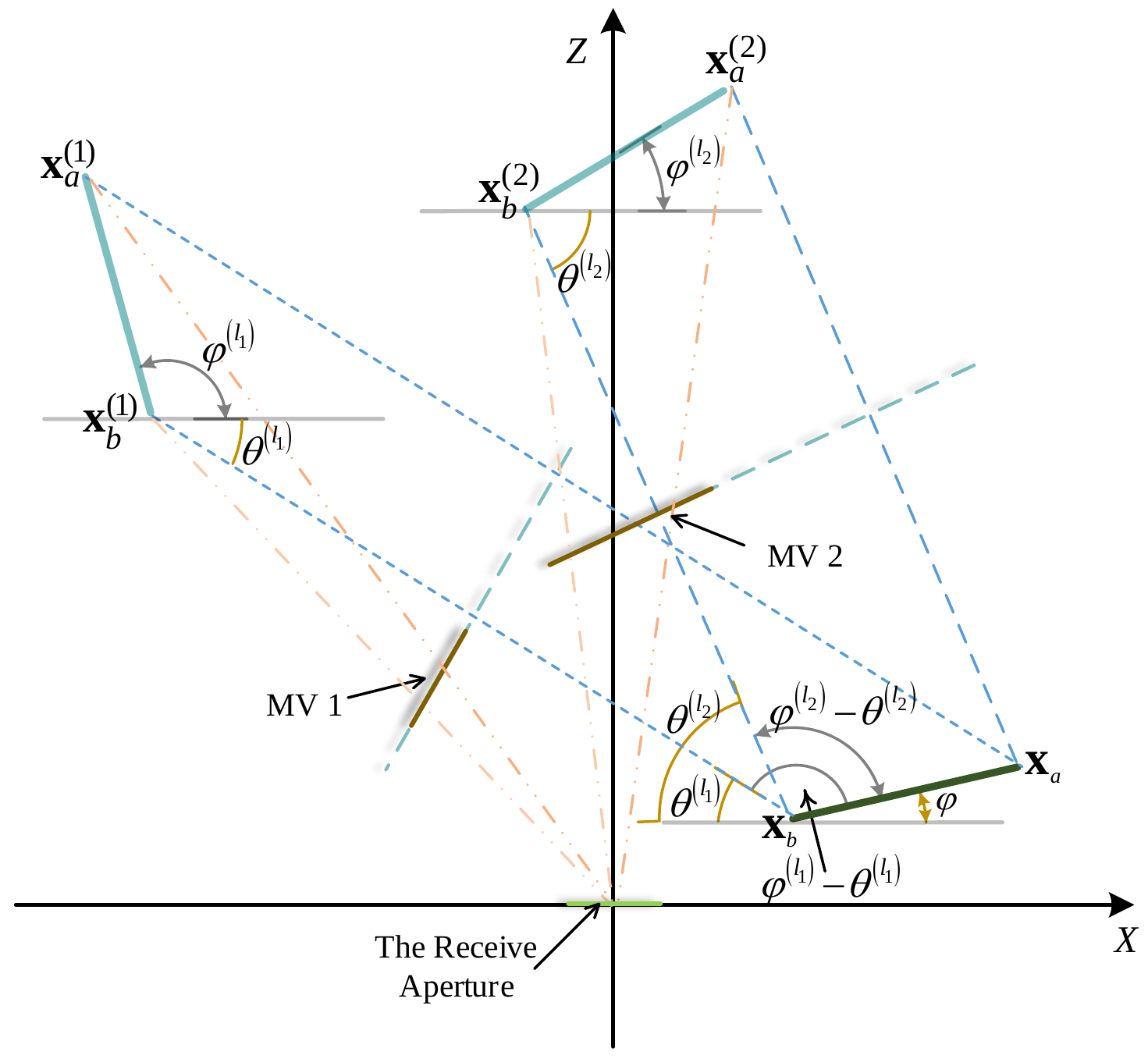}
			\caption{The geometry relations among different VPs and the TV.}\label{AngRelation}
		\end{figure}

		\subsection{Proof of Proposition \ref{Prop3} } \label{D}
		Based on the line function of MV $\ell$ surface \eqref{reflectionPlane}, as well as the symmetric geometry relation between VP $\ell$ and the TV, it is easy to establish the mathematical relation between 
		$\mathbf{x}^{(\ell)}$ and $\mathbf{x}$ as
		\begin{align}\label{locRule}
		\left\{ \begin{array}{l}
		x = {x^{(\ell )}} + \Delta {x^{(\ell )}}\\
		z = {z^{(\ell )}} + \tan ({\theta _\ell }) \cdot \Delta {x^{(\ell )}}
		\end{array} \right.,
		\end{align}
		where ${\Delta {x^{(\ell )}}}$ is the $x$-direction projection of distance between $(x,z)$ and $(x^{(\ell)},z^{(\ell)})$. The middle point of $(x,z)$ and $(x^{(\ell)},z^{(\ell)})$ locates at the MV $\ell$ surface with line function \eqref{reflectionPlane}. Thus the middle point $\left( {{x^{(\ell )}} + \frac{{\Delta {x^{(\ell )}}}}{2},{z^{(\ell )}} + \tan ({\theta _\ell }) \cdot \frac{{\Delta {x^{(\ell )}}}}{2}} \right)$ should satisfy function \eqref{reflectionPlane}. Thus ${\Delta {x^{(\ell )}}}$ can be solved by equation set expressed as
		\begin{align}\label{Deltax}
		z_{\mathsf{a}}^{(\ell )} \!\!+\! {z_{\mathsf{a}}} \!\!-\! \frac{\left(\! {{x^{(\ell )}} \!\!+\! \Delta {x^{(\ell )}} \!\!-\! x_{\mathsf{a}}^{(\ell )} \!\!-\! {x_{\mathsf{a}}}} \!\right)}{{\tan ({\theta _\ell })}}\!\!=\! {z^{(\ell )}} \!\!+\! \tan ({\theta _\ell }) \!\cdot\! \Delta {x^{(\ell )}}.
		\end{align}
		The result derived from \eqref{Deltax} is given as
		\begin{align}\label{DeriveDeltax}
		{\Delta {x^{(\ell )}}}\!=\!\left(\! {\frac{ {1 \!+\! {{\tan }^2}(\theta _\ell ^*)} }{{\tan (\theta _\ell ^*)}}} \!\right)\!\left(\!\! {\frac{{x_{\mathsf{a}}^{(\ell )} \!+\! x_{\mathsf{a}}^*\!-\!2{x^{(\ell )}}}}{{\tan (\theta _\ell ^*)}} \!+\! z_{\mathsf{a}}^{(\ell )} \!+\! z_{\mathsf{a}}^* \!-\! 2{z^{(\ell )}}} \!\!\right).
		\end{align}
		Bring \eqref{DeriveDeltax} into \eqref{locRule}, the result \eqref{FindLocation} can be obtained straightforward by replacing $(x,z)$ with the estimated value $(x^*,z^*)$. This finishes the proof.

		\bibliographystyle{ieeetr}%
		\bibliography{mirror}

		% that's all folks
	\end{document}